\documentclass[11pt]{article}
\usepackage[preprint,eandd]{neurips_2026}
\usepackage{common}

\usepackage[utf8]{inputenc}
\usepackage[T1]{fontenc}
\usepackage{amsmath,amssymb,amsfonts}
\usepackage{amsthm}
\newtheorem{proposition}{Proposition}

\usepackage{graphicx}
\usepackage{subcaption}
\usepackage{booktabs}
\usepackage{hyperref}
\usepackage{multirow}
\usepackage{tabularx}
\usepackage{natbib}
\usepackage{placeins}  
\usepackage{wrapfig}   
\usepackage{tikz}
\usetikzlibrary{matrix,positioning,arrows.meta,fit,backgrounds}
\usepackage{xspace}
\usepackage{comment}

\newcommand{\contractbench}{\textsc{ContractBench}\xspace}
\newcommand{\sha}{\operatorname{SHA\text{-}256}}

\newcommand{\repoUrl}{https://github.com/JeremyJC67/contractbench}
\newcommand{\datasetUrl}{https://huggingface.co/datasets/nips26-anon-author/contractbench}

\newcommand{\numHarborTasks}{33}
\newcommand{\numTotalTasks}{33}
\newcommand{\NumModels}{38}            
\newcommand{\NumModelFamilies}{15}     

\newcommand{\numFailureLabels}{15}
\newcommand{\numSeeds}{10}

\IfFileExists{macros/results.tex}{

}{}

\title{\contractbench: Can LLM Agents Preserve Observation Contracts?}

\author{%
  Jicheng Wang$^{1}$\quad Yifeng He$^{1\,\dagger}$\quad Zili Wang$^{1\,\dagger}$ \\[1pt]
  Hanwen Xing$^{2}$\quad Arkaprava De\quad Hao Chen$^{3\,*}$ \\[2pt]
  $^{1}$University of California, Davis \quad $^{2}$University of Southern California \\
  $^{3}$University of Hong Kong \\[2pt]
  \texttt{\{jicwang, yfhe, zliwang\}@ucdavis.edu,\ harvenx01@gmail.com,} \\
  \texttt{arkaprava08@gmail.com,\ chenho@hku.hk} \\[2pt]
  {\small $^{*}$Corresponding author. \quad $^{\dagger}$Equal contribution.}
}

\begin{document}
\maketitle

\begin{abstract}
	Tool-augmented LLM agents call APIs whose intermediate outputs, such as presigned URLs, session tokens, and OAuth state parameters, are \emph{observation contracts}: artifacts whose later use is constrained by the external system that produced them.
	We show that observation contract compliance (preserving the temporal validity and byte-level integrity) is an emergent, regression-prone capability: it is neither guaranteed by general tool-use ability nor consistently improved by larger or newer models.
	To measure this, we introduce
	\contractbench, a benchmark of \numHarborTasks{} dual-axis tasks that probe two orthogonal failure modes no existing benchmark evaluates:
	validity failures (using an artifact after expiry) and integrity failures (corrupting an artifact's bytes through the observation-to-action pipeline).
	Our evaluation is deterministic and programmatic,
	with a virtual clock controls time and SHA-256 hashes verify byte integrity.
	We assign the outcomes failure labels drawn from real-world API specifications.
	We evaluate \NumModels{} models and report four findings:
	\begin{enumerate*}[label=(\roman*)]
		\item no evaluated model clears 80\%, with Claude-Opus-4.6 leading at 77.8\%,
		revealing that current frontier models still fail to comply with observation contracts.
		\item a sharp \emph{within-family capability cliff} in Qwen~3.5 between 4B (0\%) and 9B (56.6\%), smoothing to 70.7\% at 397B-A17B\@: what emerges across the cliff is mid-trajectory restraint, not tool-call competence;
		\item \emph{non-monotonic scaling} across GPT-5 family: agentic post-training can erode compliance through sycophancy-driven regression;
		\item and our failure taxonomy works as an actionable in-context reward signal,
		yielding $+7.1$\,pp on 42 paired GPT-5.1 failures.
	\end{enumerate*}
\end{abstract}

\section{Introduction}
\label{sec:intro}

Large language model (LLM) agents are increasingly deployed as autonomous assistants by personal users and enterprises~\citep{schick2023toolformer,yao2023react}.
To reach a user-specified goal, an LLM agent calls external networked APIs,
navigates websites, executes shell commands, and reads or writes files,
observing the outcome of each step before choosing the next.
These agents are also growing more autonomous and longer-lived, with recent systems completing everyday online tasks across dozens of live platforms with minimal human oversight~\citep{zhang2026clawbench}.
In production, they already resolve software-engineering issues~\citep{yang2024sweagent}, handle multi-turn customer-service dialogues~\citep{yao2025taubench}, and carry out complex terminal tasks in scientific computing and cybersecurity~\citep{merrill2026terminalbench}.
In this setting, a wrong action can break a workflow, corrupt a database, or cause an irreversible side effect~\citep{he2025security}.
Executing reliably matters as much as understanding the task.

Several benchmarks now ground agent evaluation in real APIs, codebases, and user interactions: web navigation~\citep{zhou2024webarena}, software engineering~\citep{jimenez2024swebench,merrill2026terminalbench}, general tool use~\citep{qin2024toolllm,liu2024agentbench}, user interaction under domain-specific policies~\citep{yao2025taubench}, and multi-constraint planning~\citep{xie2024travelplanner}.
In these benchmarks, the intermediate steps of a multi-step task do not impose constraints on the steps that follow: a tool output is information for the agent's next decision, and the agent can paraphrase, summarize, or even lose it without penalty so long as the workflow reaches the right end state.
Real API workflows behave differently. The output of one step often dictates exactly how the agent must execute the next step.
For example, the agent must relay a signed URL byte-for-byte, use a token before it expires, and round-trip an OAuth state parameter unmodified.
Current benchmarks do not test whether the agent respects these inter-step constraints. Yet this is a primary failure mode in deployment: a single intermediate-step violation invalidates the whole task flow.

Previous AI agent evaluations have overlooked this setting, but it is routine in enterprise API workflows: a tool returns an artifact, and the external system that produced it constrains its later use.
We call such a tool-returned artifact an \emph{observation contract}.
Consider a presigned S3 URL returned by a cloud-storage API.
The URL works only before its expiration time and only if the agent preserves its signature-bound query string byte-for-byte~\citep{aws2026presignedurls}.
An agent can understand the user's goal, call the right API, and still fail the workflow by using the URL too late, truncating it, re-encoding reserved characters, or reordering signed parameters.
This failure pattern is not unique to presigned URLs. OAuth state parameters, session tokens, webhooks authenticated with HMAC (hash-based message authentication code), and rate-limit windows all impose the same two constraints: \emph{temporal validity} (use the artifact before it expires or becomes stale) and \emph{byte-level integrity} (relay the artifact without alteration).
The two constraints are orthogonal. Refreshing an expired artifact repairs validity but does not prevent in-transit corruption of every fresh copy; preserving bytes intact does not teach the agent to respect deadlines or version freshness.

\begin{wraptable}[13]{r}{0.55\textwidth}
	\vspace{-10pt}
	\centering
	\caption{%
		Agent benchmark landscape.
		``V'': temporal validity;
		``I'': byte-level integrity;
		``P'': programmatic evaluation (no human or LLM judge).
	}
	\label{tab:landscape}
	\footnotesize
	\begin{tabular}{@{}lccc@{}}
	\toprule
	\textbf{Benchmark}                               & \textbf{V} & \textbf{I} & \textbf{P} \\
	\midrule
	SWE-bench~\citep{jimenez2024swebench}            & \xmark     & \xmark     & \cmark     \\
	ToolBench~\citep{qin2024toolllm}                 & \xmark     & \xmark     & Partial    \\
	AgentBench~\citep{liu2024agentbench}             & \xmark     & \xmark     & \cmark     \\
	TravelPlanner~\citep{xie2024travelplanner}       & \xmark     & \xmark     & Partial    \\
	$\tau$-bench~\citep{yao2025taubench}             & \xmark     & \xmark     & \cmark     \\
	Terminal-Bench~\citep{merrill2026terminalbench}  & \xmark     & \xmark     & \cmark     \\
	TicToc~\citep{cheng2026llmagentstemporallyblind} & Partial    & \xmark     & \cmark     \\
	\midrule
	\contractbench (Ours)                            & \cmark     & \cmark     & \cmark     \\
	\bottomrule
\end{tabular}

	\vspace{-0.8em}
\end{wraptable}

Existing benchmarks partly address the temporal axis and largely overlook the integrity axis.
\autoref{tab:landscape} surveys the landscape: several include partial temporal constraints such as time limits or session windows, and recent work on temporal blindness shows that agents can mishandle stale information~\citep{cheng2026llmagentstemporallyblind}.
None, however, jointly tests temporal validity and byte-level artifact integrity.
In practice, expired tokens, mutated URLs, stale ETags (HTTP entity tags), and signature mismatches break otherwise correct workflows.

To address this gap, we introduce \contractbench{},
a benchmark for measuring observation-contract compliance in LLM agents.
\contractbench{} contains \numHarborTasks{} tasks drawn from real API contract patterns --- 
presigned URLs, OAuth state parameters, signed requests, HMAC-protected webhooks, rate-limit windows,
and multi-step token workflows --- 
arranged across the two orthogonal axes and eight contract domains, shown in \autoref{fig:overview}.
Every task runs under a deterministic virtual clock, validates artifacts programmatically, and assigns each failed episode one primary label from a taxonomy of \numFailureLabels{} failure modes that map directly onto either the validity or the integrity axis.
The benchmark measures both whether an agent reaches a goal and whether it preserves the contracts attached to intermediate tool outputs.
Across \NumModels{} model variants from \NumModelFamilies{} families, observation-contract compliance is unsolved at the frontier (best model 77.8\,\% on the full \numHarborTasks-task suite, $n{=}99$) and behaves non-monotonically with scale and post-training: a sharp within-family capability cliff in Qwen~3.5 and a V-shape across GPT-5 generations at near-constant parameter count.
Our contributions are:
\begin{enumerate}[leftmargin=*]
	\item \textbf{Observation contracts.} We formalize observation contracts as tool-returned artifacts governed by two orthogonal constraints, temporal validity and byte-level integrity, and define a taxonomy of \numFailureLabels{} failure modes for diagnosing contract failures.
	\item \textbf{\contractbench.} We build a deterministic benchmark of \numHarborTasks{} agent tasks drawn from real API contract patterns, covering validity-dominant, integrity-dominant, and dual-axis pressures. Each task has a programmatic validator, a reproducible virtual-clock execution model, and a reference solution that establishes solvability.
	\item \textbf{Empirical findings.} Across \NumModels{} model variants spanning \NumModelFamilies{} families, contract compliance is neither guaranteed by general tool-use ability nor consistently improved by larger or newer models. We observe a sharp within-family capability cliff in Qwen~3.5 (0\% at ${\leq}4$B to 70.7\% at 397B-A17B), a GPT-version regression and recovery (23\% $\to$ 71\% $\to$ 49\% $\to$ 75\%, full 33-task suite, $n{=}99$), and a structured failure hierarchy from protocol entry to long-horizon constraint tracking.
\end{enumerate}

\begin{figure}[!htbp]
	\centering
	\includegraphics[width=\textwidth]{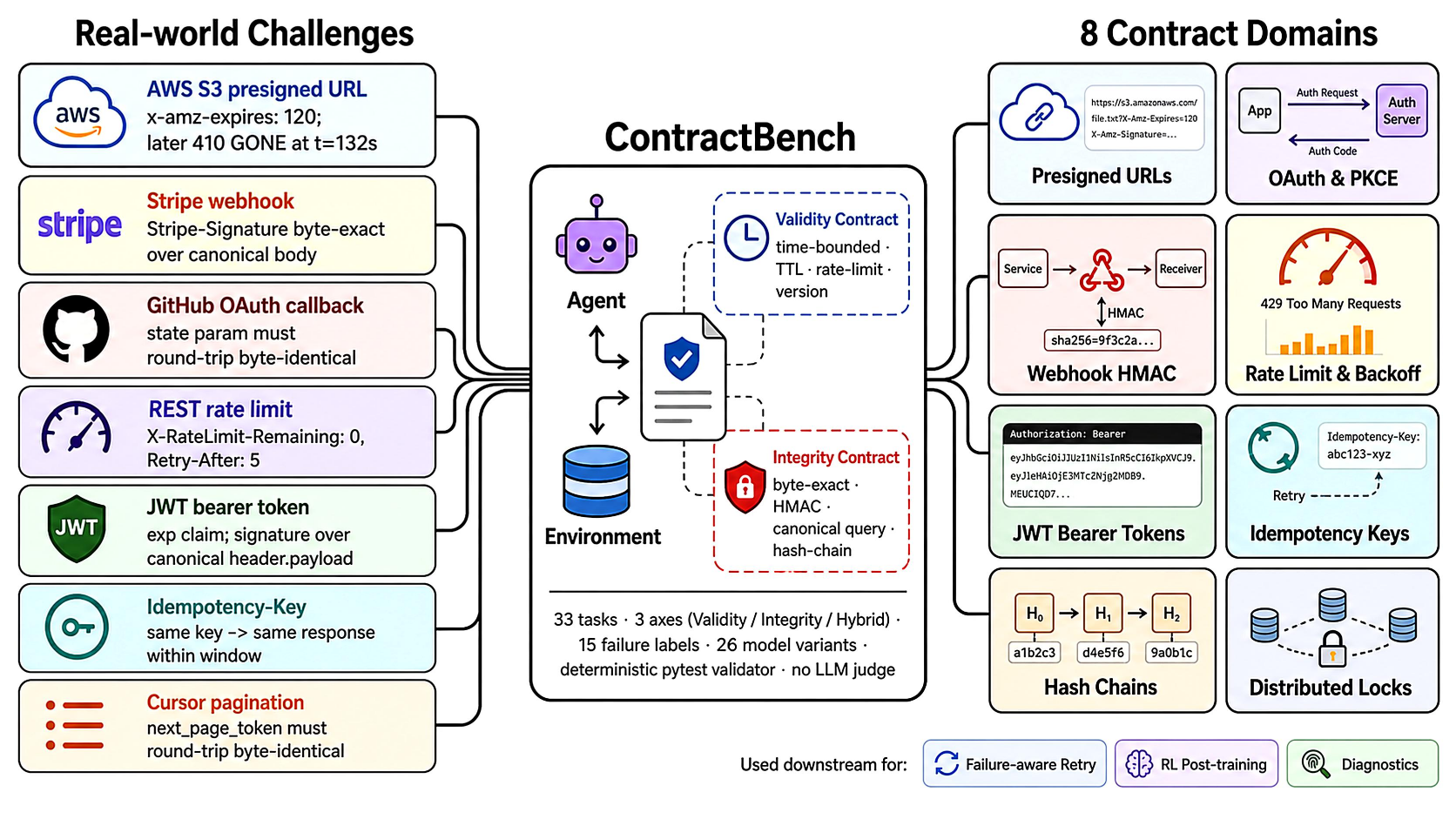}
	\caption{\contractbench{} is the first deterministic benchmark for LLM-agent observation-contract compliance, spanning \numHarborTasks{} tasks across two orthogonal axes (temporal validity and byte-level integrity) and eight real-world API contract domains.}
	\label{fig:overview}
\end{figure}

\section{Observation Contracts}
\label{sec:problem}

Production agents repeatedly hit the same class of failure:
a tool returns an \emph{artifact} (an intermediate value such as a presigned URL, a signed token, or an OAuth state parameter that the agent must reuse in a later step) that carries a time-limit and a byte-level integrity rule,
and a downstream tool call rejects it because one of those rules was violated somewhere along the way~\citep{aws2026presignedurls}.
We call this an \emph{observation contract} failure,
which we adopt from the enforcement clauses of widely deployed API specifications (AWS SigV4, OAuth~2.0, JWT, Stripe webhooks, RFC~6585), where every contract clause we observed reduces to a temporal check, a byte-level check, or a conjunction of the two.
Despite its prevalence in deployed agent stacks, the agent-evaluation literature has not formalized this class or measured agents against it directly.
We provide a formal definition of the problem in this section,
and introduce \contractbench in \autoref{sec:benchmark}.


Let $\mathcal{O}$ be an artifact space (\eg, the set of UTF-8 strings), $\mathcal{T} \subseteq \mathbb{R}_{\geq 0}$ a virtual time domain, and $\mathcal{H}$ a digest space (in our implementation, $\mathcal{H} = \{0,1\}^{256}$ via SHA-256).

\begin{definition}[Observation contract]\label{def:contract}
	An \emph{observation contract} is a 4-tuple
	$C \;=\; (o,\; t_\text{issue},\; \tau,\; \pi)$,
	where $o \in \mathcal{O}$ is the artifact issued at time $t_\text{issue} \in \mathcal{T}$, $\tau \in \mathbb{R}_{>0}$ is its time-to-live, and $\pi : \mathcal{O} \to \{0,1\}$ is an \emph{integrity predicate}. The induced \emph{validity window} is $W(C) := [\,t_\text{issue},\; t_\text{issue} + \tau\,) \subseteq \mathcal{T}$.
\end{definition}

In our implementation $\pi$ is the SHA-256 equality check $\pi(o') = \mathbbm{1}[\,\sha(o') = h(o)\,]$, but the framework admits any deterministic predicate (HMAC, ETag equality, JWT signature verification).

\begin{definition}[Satisfaction]\label{def:sat}
	A submission pair $(o',\, t') \in \mathcal{O} \times \mathcal{T}$ satisfies a contract $C$ if and only if
	\[
		\underbrace{t' \in W(C)}_{\text{validity}}
		\quad\text{and}\quad
		\underbrace{\pi(o') = 1}_{\text{integrity}}.
	\]
	We write $\operatorname{Sat}(C, o', t') \in \{0,1\}$ for the satisfaction indicator.
\end{definition}

\begin{definition}[Failure modes]\label{def:fail}
	A submission $(o', t')$ is a \emph{validity failure} on $C$ when $t' \notin W(C)$, and an \emph{integrity failure} on $C$ when $\pi(o') = 0$. The two modes are logically independent: a submission may exhibit one, both, or neither.
\end{definition}

\begin{definition}[Compliance]\label{def:trace-compliance}
	A trace of $n$ contracts $\{C_i\}_{i=1}^n$ with submissions $\{(o'_i, t'_i)\}_{i=1}^n$ is contract-compliant if and only if
	\[
		\operatorname{Compliant}(\{C_i\}, \{(o'_i, t'_i)\}) \;:=\; \prod_{i=1}^{n} \operatorname{Sat}(C_i, o'_i, t'_i) \;=\; 1.
	\]
	Equivalently, the trace is non-compliant whenever $\operatorname{Sat}(C_i, o'_i, t'_i) = 0$ for some $i \in \{1, \dots, n\}$.
\end{definition}


\begin{proposition}[Orthogonality]\label{prop:orthogonal}
	For any contract $C$ with non-empty validity window and non-trivial integrity predicate, all four cells of the $2{\times}2$ validity-by-integrity partition over submissions are non-empty (\autoref{fig:concept}). Consequently, fixing one axis (\eg, always submitting at $t' \in W$) does not fix the other. (Proof sketch in \autoref{app:orthogonal-proof}.)
\end{proposition}

The proposition is a structural claim: the two axes can vary independently in principle. The substantive claim is that they vary independently in practice, in real LLM agent stacks. Refreshing an expired URL repairs validity but does not stop a downstream pipeline from corrupting the fresh copy through truncation, re-encoding, or parameter reordering. Conversely, preserving bytes intact does not teach an agent to respect deadlines or rate-limit windows. \autoref{sec:experiments} confirms this empirically: models near the frontier on one axis collapse on the other.
This independence is what makes the problem worth a dedicated benchmark: a single-axis evaluation cannot expose it. The next section describes how \contractbench probes both axes, separately and jointly.

\section{\contractbench{}}
\label{sec:benchmark}

We introduce \contractbench, 
a benchmark that instantiates the observation-contract problem (\autoref{sec:problem}) as a suite of programmatically evaluated tasks.
\contractbench contains \numHarborTasks{} tasks, 
all generated from parameterized templates, evaluated by deterministic validators, and executed under a virtual clock for full reproducibility. 
This section describes the dual-axis task design (\autoref{sec:tasks}) following \autoref{def:sat},
and the deterministic evaluation protocol (\autoref{sec:protocol}).

\subsection{Building Tasks that Probe Both Axes Simultaneously}
\label{sec:tasks}

\begin{wrapfigure}[20]{r}{0.55\textwidth}
	\centering
	\vspace{-10pt}
	\includegraphics[width=0.53\textwidth]{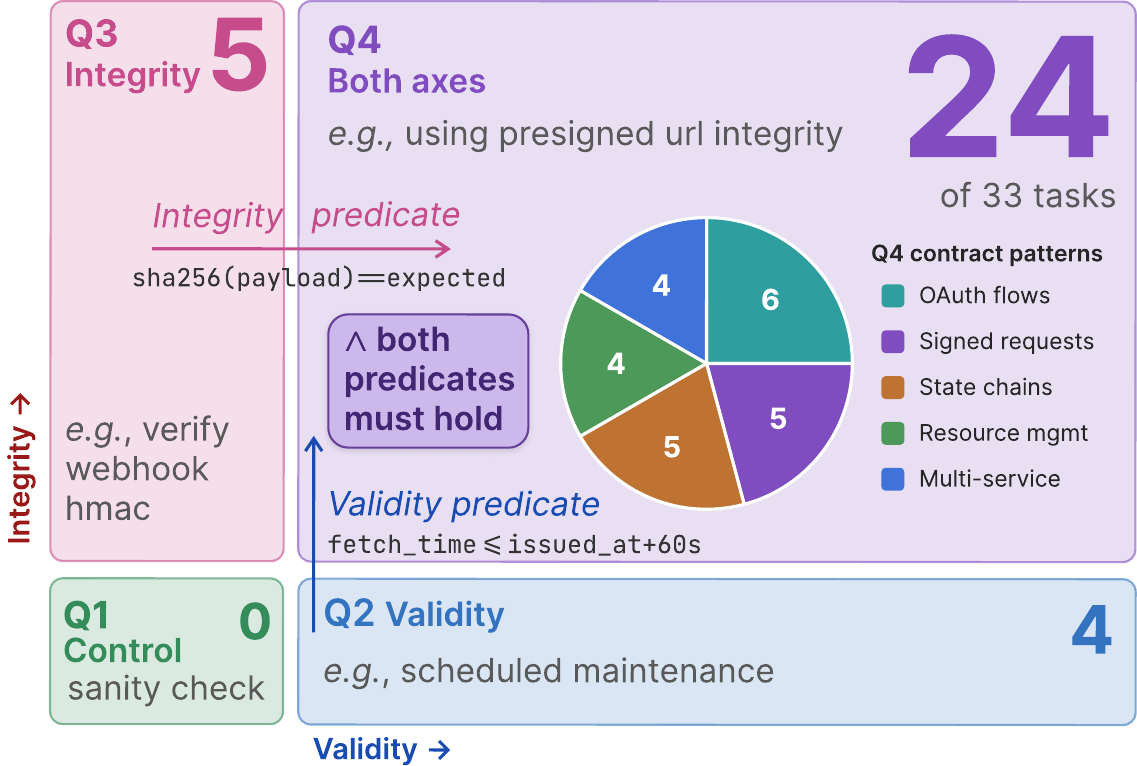}
	\caption{\textbf{Two orthogonal predicates define an observation contract.} Integrity ($\sha(\text{payload})\!=\!\text{expected}$) and validity ($t_\text{fetch}\!\leq\!t_\text{issue}{+}\tau$) must \emph{both} hold (\autoref{def:sat}); the four quadrants instantiate \autoref{prop:orthogonal}. Quadrant area encodes \contractbench's task distribution.}
	\label{fig:concept}
\end{wrapfigure}

\paragraph{Two-level taxonomy.}
\label{sec:task_families}
We organize tasks using a two-level taxonomy, as depicted in \autoref{fig:concept}.
The top level groups tasks by primary pressure quadrant following the orthogonality of \autoref{prop:orthogonal}:
\begin{enumerate*}[label=\textbf{Q\arabic*}]
\item (low validity, low integrity) is a sanity check;
\item isolates temporal planning;
\item isolates byte preservation; and
\item stresses both axes simultaneously.
\end{enumerate*} The Q4-heavy distribution (24 of \numHarborTasks{} tasks) reflects the real-world prevalence of artifacts that are both time-limited and signature-bound.
The second level (visible inside Q4) groups dual-axis tasks by \emph{contract pattern}: the real-world API pattern being tested (OAuth / auth flows, signed requests, state chains, resource management, multi-service flows). Per-task counts and one example task per pattern are shown in the figure; the full per-task list is in \appautoref{app:tasks}.

\begin{table}[t]
	\centering
	\caption{Failure taxonomy (\numFailureLabels{} labels). Each failed episode receives one primary label, aggregated from the underlying event list using the most-severe rule (\appautoref{app:rl_reward}).}
	\label{tab:taxonomy}
{\scriptsize
	\setlength{\tabcolsep}{4pt}
	\renewcommand{\arraystretch}{1.10}
	\begin{tabular*}{\textwidth}{@{\extracolsep{\fill}}llp{10.5cm}@{}}
		\toprule
		\textbf{Group} & \textbf{Label}                  & \textbf{Description}                                                                                \\
		\midrule
		Validity       & \texttt{EXPIRED\_BEFORE\_USE}   & Resource TTL expired before agent used it~\citep{rfc6750}.                                          \\
		& \texttt{RATE\_LIMITED}          & Exceeded API rate-limit window (HTTP 429)~\citep{rfc6585}.                                          \\
		& \texttt{SCHEDULED\_UNAVAILABLE} & Called service during scheduled downtime (HTTP 503)~\citep{rfc9110}.                                \\
		& \texttt{VERSION\_CONFLICT}      & Submitted write with stale ETag / If-Match~\citep{rfc7232}.                                         \\
		\midrule
		Integrity      & \texttt{WRONG\_VALUE}           & Required field present but value incorrect (benchmark-internal predicate; see \autoref{sec:tasks}). \\
		& \texttt{MISSING\_CONSTRAINT}    & Required protocol field omitted from request~\citep[\S4.1.1]{rfc6749}.                              \\
		& \texttt{MUTATED\_TOKEN}         & Bytes altered by URI normalization/re-encoding /reordering~\citep[\S6]{rfc3986}.                     \\
		& \texttt{SIGNATURE\_MISMATCH}    & Request signature failed verification~\citep{rfc9421,aws2024sigv4}.                                 \\
		& \texttt{COMPENSATION\_FAILURE}  & Saga compensation step did not undo prior write~\citep{garciamolina1987sagas}.                      \\
		& \texttt{WRONG\_HASH}            & Content digest computed incorrectly~\citep{nist2015fips180,rfc9530}.                                \\
		& \texttt{SHORTCUT\_TAKEN}        & Adversarial shortcut endpoint used instead of canonical flow~\citep{krakovna2020specification}.     \\
		& \texttt{MISSING\_TOKEN}         & Required authentication token absent at submission~\citep[\S4]{rfc6749}.                            \\
		& \texttt{BACKOFF\_VIOLATION}     & Retried before \texttt{Retry-After} window elapsed~\citep[\S10.2.3]{rfc9110}.                       \\
		\midrule
		Meta           & \texttt{SUCCESS}                & Contract satisfied on both axes.                                                                    \\
		& \texttt{OTHER}                  & Unclassified or low-frequency variant (catch-all for emitted strings outside the taxonomy).         \\
		\bottomrule
	\end{tabular*}
}

\end{table}

\paragraph{Task structure.}
\label{sec:files}
Each \contractbench{} task consists of four files: 
\begin{enumerate*}
\item a TOML metadata header that fixes difficulty, dual-axis category, and timing budgets; 
\item a Markdown instruction visible to the agent; 
\item a FastAPI server that emits one of \numFailureLabels{} deterministic failure labels per HTTP request via a shared \texttt{append\_log()} call running over a virtual clock; 
\item a \texttt{pytest} validator that consumes the request log and computes the reward. 
\end{enumerate*}
The agent only sees the instruction and the server's HTTP responses; 
the metadata header and validator are hidden, eliminating test-set leakage and LLM-as-judge ambiguity. 
We build \contractbench{} on Harbor~\citep{harborframework} following standard practice of agent evaluation~\citep{merrill2026terminalbench,li2026skillsbench}.
We provide an annotated example in \appautoref{fig:task_anatomy} to illustrate the task structure.

\paragraph{Failure taxonomy.}
\label{sec:taxonomy}
\contractbench{} instantiates the abstract validity/integrity axes as a \numFailureLabels-label taxonomy (\autoref{tab:taxonomy}), 
partitioned into four \emph{validity} failures (temporal), 
nine \emph{integrity} failures (byte-level), and two \emph{meta} labels. 
Each label in \autoref{tab:taxonomy} traces to a documented real-world API rule.
We provide the full per-task event-to-label mapping in \appautoref{app:rl_reward}. 
At runtime, each task server emits a chronological list of violation events into the request log; 
the validator aggregates these events into a single \emph{primary label} per episode using the most-severe-label rule, 
and protocol state events that appear in successful episodes (\eg, \texttt{CHALLENGE\_ISSUED}) are excluded from the failure distribution.


Two sources of difficulty stack in \contractbench{}. 
First, the path from API response, through the agent's context window, 
into the next tool call is not byte-preserving in production: 
tool-call interfaces truncate long inputs, HTTP libraries re-encode percent-escaped URLs, 
middleware re-sorts query parameters before signing, 
and rendered link text disagrees with the underlying \texttt{href}. 
These mutations are not artifacts we invent: they arise naturally from the way real agent stacks compose.
\contractbench{} reproduces five of them as deterministic toggles so each cause can be isolated,
and the full catalog and the failure label each one triggers is in \appautoref{app:path_mutations}. 
Furthermore, even without mutations, 
dual-axis tasks force the agent to plan around tight TTLs while preserving exact byte content.
\appautoref{app:examples} walks through one validity-heavy and one integrity-heavy task that exemplify the two pressures and the corresponding agent-side strategies.

\paragraph{Task generation and solvability.}
\label{sec:oracle}
Each task is instantiated from a parameterized template that fixes the contract type, 
the difficulty knobs (TTL duration, rate-limit window, token length), and the validator. 
A single \emph{seed} controls every random element (resource ordering, token content, timing jitter), 
so a given (template, seed) pair always produces the same task instance, 
and we evaluate each model on a fixed grid of \numSeeds{} seeds per task. 
Following the verification practice of Terminal-bench,
every \contractbench{} task ships with an \emph{oracle solution}
that issues the correct sequence of API calls under the same virtual clock as the agent.
We require each oracle to pass the validator before admitting the task; 
all \numHarborTasks{} oracles do so on every seed ($\numHarborTasks{}\times\numSeeds{}=330$ oracle episodes), 
so every task is provably solvable and the validator accepts the intended behavior. 
Oracle solutions serve as regression tests during task authoring and reproducibility checks. 
They are not exposed to the agent during model evaluation.

\subsection{Evaluation Protocol}\label{sec:protocol}

\contractbench{} uses deterministic, programmatic validators for reliability: no LLM-as-judge, no human raters.
The validator consumes the per-episode HTTP request log written by the task server and outputs the reward details:
a binary \texttt{success} flag, 
the most-severe failure label drawn from the taxonomy (\autoref{tab:taxonomy}), 
the failure detail, 
and trace metadata (\appautoref{app:schema}).

\paragraph{Metrics.}
Each model is evaluated at $k{=}3$ runs per task across \numHarborTasks{} tasks ($n{=}99$ episodes per model).
We report three quantities derived from the validator output.
\textbf{Success rate} (SR\,\%) is the fraction of the $n{=}99$ episodes with label \texttt{SUCCESS}.
\textbf{Per-task pass rate} is the mean reward across the $k{=}3$ runs of one model on one task.
\textbf{Failure-label distribution} is the distribution of primary labels over a model's failed episodes (counting convention in \appautoref{app:reproducibility}).

\paragraph{Success criteria.}
An episode is labeled \texttt{SUCCESS} if and only if its HTTP request log satisfies the contract per \autoref{def:sat}: every required observation is both temporally valid ($t_\text{fetch}\!\leq\!t_\text{issue}+\tau$) and byte-integral ($\sha(\text{payload})\!=\!\text{expected}$); any other event in the log triggers the most-severe failure label.
The validator reads HTTP-level events only, never agent prose or self-reported summaries, so an agent cannot bluff its way to success by narrating compliance it did not perform.

\section{Experiments}
\label{sec:experiments}


\paragraph{Models and settings.}
We evaluate \NumModels{} model variants across \NumModelFamilies{} families, grouped into three categories:
\begin{enumerate*}[label=(\roman*)]
	\item frontier proprietary models (Claude, GPT-5 series, and Gemini 2.5);
	\item open-source instruction-tuned models (Qwen 3.5, Qwen 2.5, MiniMax, Mistral / Ministral, Gemma, Llama, DeepSeek-R1) spanning a range of parameter counts and training recipes;
	\item Base (non-instruction-tuned) and smaller Instruct checkpoints, included to disentangle parameter count from post-training.
\end{enumerate*}
Per-model provider and infrastructure metadata, plus the full \NumModels-row leaderboard, are in \appautoref{tab:llm_results}.
Each model is run on the full \numHarborTasks-task \contractbench{} suite at $k{=}3$ rollouts per task ($n{=}99$ episodes per model), with temperature 0, a task-describing system prompt, and a 600-second per-episode wall-clock timeout.

\begin{wraptable}[24]{r}{0.4\textwidth}
	\vspace{-0.8em}
	\centering
	\caption{%
		The \contractbench{} leaderboard.
		$k{=}3$, $n{=}99$ episodes per model.
	}
	\label{tab:llm_results_frontier}
	\footnotesize
	\begin{tabular}{@{}lrr@{}}
	\toprule
	\textbf{Model}           & \textbf{Pass} & \textbf{SR\%} \\
	\midrule
	Claude Opus 4.6          & 77            & \textbf{77.8}                    \\
	Claude Sonnet 4.5        & 69            & 69.7                             \\
	\lightmidrule
	GPT-5.2                  & 74            & 74.8                             \\
	GPT-5                    & 70            & 70.7                             \\
	GPT-5.1                  & 48            & 48.5                             \\
	GPT-4o                   & 23            & 23.2                             \\
	\lightmidrule
	Gemini 2.5 Pro           & 51            & 51.5                             \\
	Gemini 2.5 Flash         & 19            & 19.2                             \\
	\midrule
	Qwen3.5-397B-A17B (MoE)  & 70            & 70.7                             \\
	Qwen3.5-27B              & 64            & 64.6                             \\
	Qwen3.5-9B               & 56            & 56.6                             \\
	Qwen2.5-72B-Instruct     & 23            & 23.2                             \\
	\lightmidrule
	MiniMax-M2.5             & 62            & 62.6                             \\
	MiniMax-M2.1             & 60            & 60.6                             \\
	MiniMax-M2               & 53            & 53.5                             \\
	\lightmidrule
	Mistral-Small-4 (MoE)    & 42            & 42.4                             \\
	Ministral-3-14B          & 28            & 28.3                             \\
	\lightmidrule
	Gemma-4-26B-A4B (MoE)    & 38            & 38.4                             \\
	Gemma-4-31B              & 37            & 37.4                             \\
	\lightmidrule
	Llama-3.3-70B-Instruct   & 7             & 7.1                              \\
	\bottomrule
\end{tabular}

	\vspace{-0.8em}
\end{wraptable}

\paragraph{Research questions.}
To probe how contract compliance varies with models, scale, post-training, and failure mode,
we ask four questions. 
\begin{enumerate*}[label=\textbf{RQ\arabic*}]
	\item how do frontier models perform on \contractbench{}? 
	\item how does compliance scale with parameters within the same training method, and does the within-family pattern reproduce across families?
	\item at approximately constant parameter count, can post-training updates change compliance, and if so, on which capability axes?
	\item does \contractbench{} reward signal lead to improvement of the agents?
\end{enumerate*}

\FloatBarrier
\subsection{Frontier Models on Observation Contracts}
\label{sec:frontier}

The \contractbench{} leaderboard (\autoref{tab:llm_results_frontier}) places Claude Opus 4.6 at the top (77.8\,\%, 77/99); no evaluated model clears 80\,\%, so roughly one in four episodes still violates the contract even for the top model (per-task ceiling in \autoref{sec:labels}).
Frontier proprietary scores span $19.2$\,\%--$77.8$\,\% (a 58.6-pp within-tier spread), and the best-open vs.\ best-proprietary gap is only $7.1$\,pp: Qwen3.5-397B-A17B matches GPT-5 at $70.7$\,\%.
Rank on \contractbench{} does not track rank on general-purpose benchmarks --- Gemini 2.5 Flash and GPT-5.1 sit below several open-source SOTA models, while Llama-3.3-70B ($7.1$\,\%) trails far behind same-class proprietary models.
The GPT-5 series moves non-monotonically across versions (\autoref{sec:version_regression}), and the open-source distribution falls off steeply at the bottom (\autoref{sec:cliff}).

\begin{finding}{Frontier ceiling and cross-provider disagreement}
	\label{find:frontier_ceiling}
	No evaluated model clears 80\,\%, and frontier proprietary scores span $19.2$\,\%--$77.8$\,\% (Claude Opus 4.6 leads). Rank on \contractbench{} does not track general-purpose benchmark rank: Gemini 2.5 Flash and GPT-5.1 sit below several open-source SOTA models, and the 7.1-pp best-open vs.\ best-proprietary gap is narrower than the within-proprietary spread.
\end{finding}

\FloatBarrier
\subsection{Within-Family Scaling: A Capability Cliff}
\label{sec:cliff}
\label{sec:scaling}

\begin{figure}[t]
	\centering
	\includegraphics[width=\textwidth]{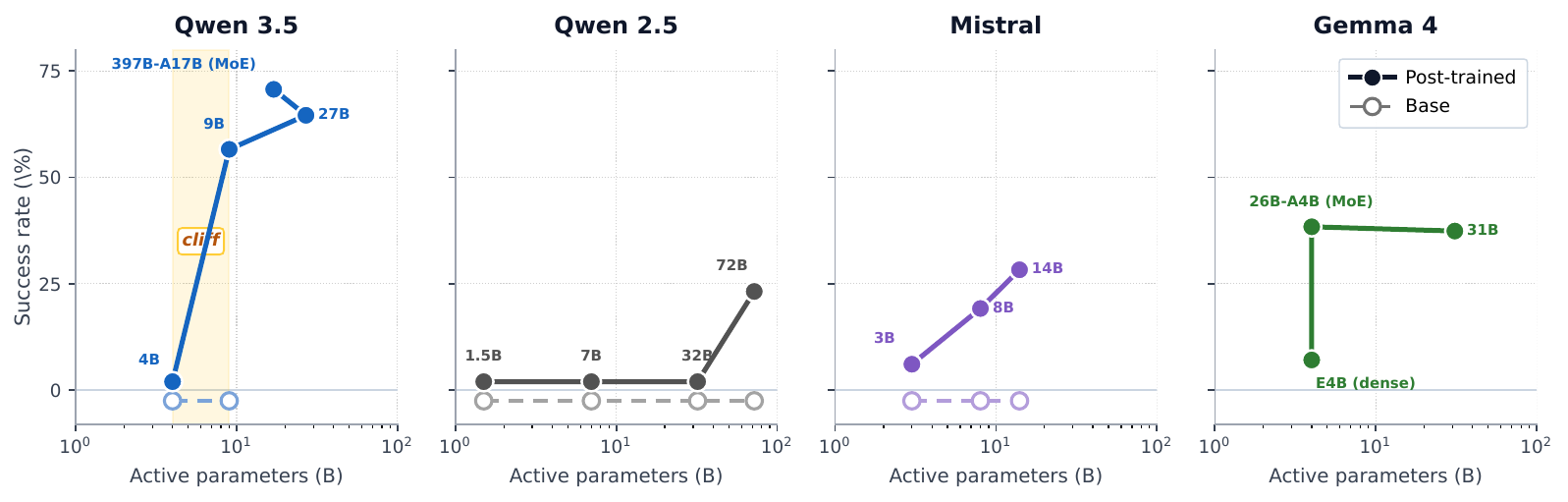}
	\caption{Within-family scaling on \contractbench{}: at fixed parameter count, the Base$\to$Instruct gap (post-training delta) is family-specific, with Qwen 3.5 reaching 56.6\,\% already at 9\,B-Instruct.}
	\label{fig:scaling}
\end{figure}

\paragraph{The Qwen 3.5 cliff.}
We analyze the impact of scaling on observation contracts with the Qwen~3.5 dense family in \autoref{fig:scaling},
and observe a sharp jumps on success rate between 4\,B and 9\,B.
Below the cliff, the 4\,B agent emits valid tool calls but gives up after one or two failed attempts, so every episode terminates before any contract predicate fires (100\,\% \texttt{OTHER}, a pre-contract failure rather than a taxonomy gap). Above it, traces show adaptive behavior: agents pivot between \texttt{curl} and \texttt{python3 urllib}, wait through scheduled maintenance, and back off after rate limits, and their residual failures move into frontier-style labels (\texttt{MISSING\_CONSTRAINT}, temporal).
What comes online is mid-trajectory restraint, not tool-call emission.

\paragraph{Other families do not reproduce the cliff.}
Qwen 2.5 Instruct climbs gradually instead of jumping (0\,\% at 1.5--32\,B, 23.2\,\% only at 72\,B). Mistral / Ministral rises smoothly (Ministral-3-3B/8B/14B at 6.1/19.2/28.3\,\%; Mistral-Small-4 MoE 42.4\,\%). Gemma 4 ramps early and then plateaus (E2B 7.1\,\%, E4B 17.1\,\%, 26B-A4B 38.4\,\%, 31B 37.4\,\%); E4B and 26B-A4B share 4\,B active compute, so the dense-vs-MoE gap reads as a vertical jump in the panel. Phi-4 14\,B Base, Llama-3.3-70B-Instruct, and DeepSeek-R1 671\,B/37\,B (Thinking MoE) all sit at or near 0\,\%, so chain-of-thought training alone does not lift compliance. The abruptness inside Qwen 3.5 fits the emergent-capability pattern of~\citet{wei2022emergent}, but the threshold tracks training recipe more than scale.

\begin{finding}{Capability cliff: emergence is family-specific, not parameter-fixed}
	\label{find:capability_cliff}
	Qwen 3.5 Instruct moves from 0\,\% at 4\,B to 56.6\,\% at 9\,B and scales smoothly to 70.7\,\% at 397\,B-A17\,B. The cliff localizes to the Qwen 3.5 training recipe rather than parameter count, and what turns on across it is adaptive recovery, not tool-call emission.
\end{finding}

\paragraph{Contract compliance emerges from post-training.}
The Base$\to$Instruct contrast in \autoref{fig:scaling} is consistent across the families for which Base weights are publicly available. Qwen 2.5 Base scores 0\,\% across 1.5--72\,B and Mistral / Ministral Base does the same across 3--14\,B; the Qwen 3.5 9\,B Base$\to$9\,B Instruct comparison reproduces the gap inside the recipe that drives the cliff (Phi-4 14\,B Base also at 0\,\%). We could not run the matching test for Qwen 3.5 27\,B / 397\,B-A17\,B or Gemma 4, since their Base checkpoints are not publicly released. Wherever the comparison is possible, no Base checkpoint clears the floor while the matched Instruct version does. Pretraining scale alone is therefore insufficient: contract compliance is acquired during post-training. \autoref{sec:version_regression} shows the converse direction, that post-training updates at fixed parameter count can also erode this capability.

\begin{finding}{Contract compliance is a post-training capability}
	\label{find:posttraining_capability}
	Frontier-tier compliance requires reasoning post-training or at least instruction-tuning on top of the base model.
	Pretraining scale alone does not produce contract compliance,
	and (\autoref{sec:version_regression}) post-training can also remove it.
\end{finding}

\FloatBarrier
\subsection{Post-Training Regression: A V-Shape in GPT-5}
\label{sec:version_regression}

\begin{figure}[t]
	\centering
	\begin{subfigure}[t]{0.42\textwidth}
		\includegraphics[width=\linewidth]{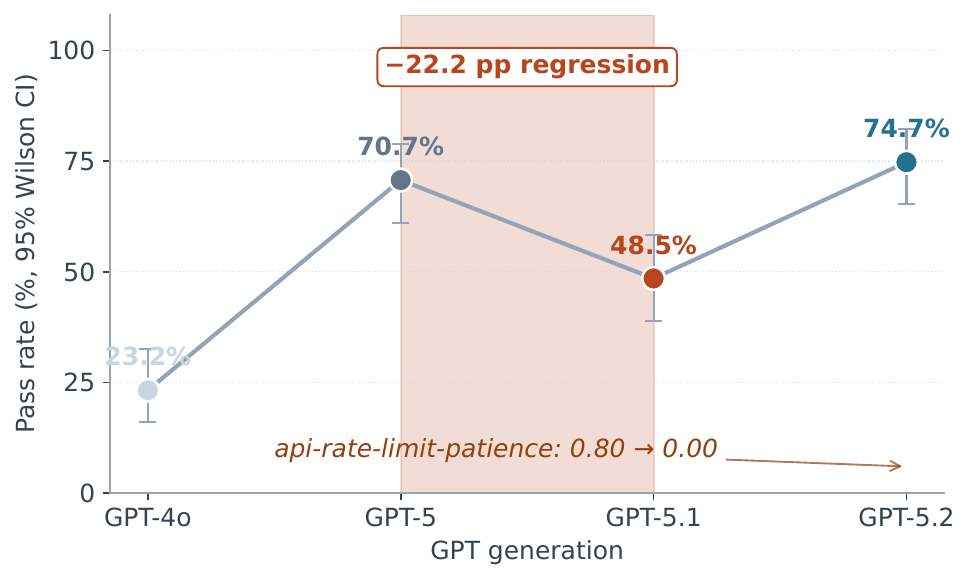}
		\caption{V-trajectory across four GPT generations.}
		\label{fig:v_shape_a}
	\end{subfigure}\hfill
	\begin{subfigure}[t]{0.28\textwidth}
		\includegraphics[width=\linewidth]{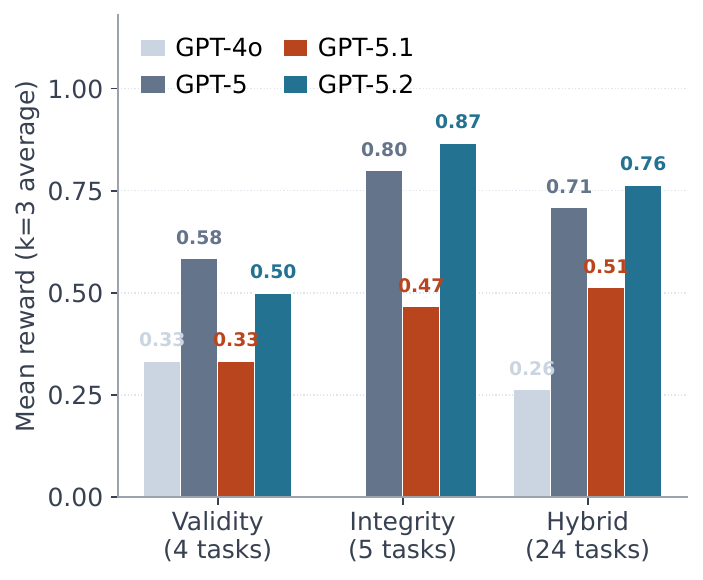}
		\caption{Per-family mean reward.}
		\label{fig:v_shape_b}
	\end{subfigure}\hfill
	\begin{subfigure}[t]{0.28\textwidth}
		\includegraphics[width=\linewidth]{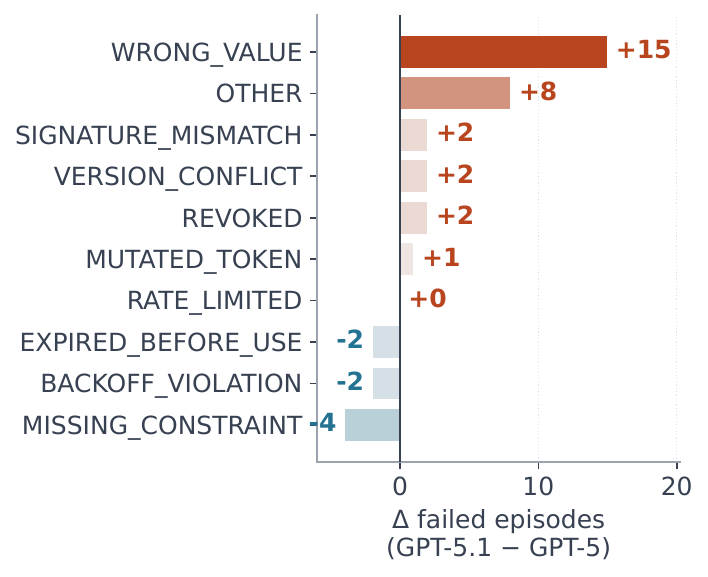}
		\caption{Failure label differences.}
		\label{fig:v_shape_c}
	\end{subfigure}
	\caption{A V-shape performance within the GPT-5 family with shared pretrained base: a structured, locatable post-training version regression. The regression is structured (specific failure modes return) and locatable (concentrated on the integrity axis), not a uniform capability loss.}
	\label{fig:v_shape}
\end{figure}

In \autoref{find:posttraining_capability}
we observe that contract compliance emerges from post-training.
Holding the pretrained base fixed, we further investigate the converse question:
how does scaling post-training on the same base affect contract compliance?
We experiment on the GPT-5 series,
where GPT-5, GPT-5.1, and GPT-5.2 share a pretrained base and differ only in post-training updates~\citep{openai2025gpt51,openai2025gpt51dev,openai2026gpt55}.
Across these generations, compliance on the full 33-task suite is non-monotonic (\autoref{fig:v_shape_a}).
We observe a \emph{version regression} on GPT-5.1 driven entirely by post-training update choices,
not model size or pretraining changes.

\paragraph{The regression is structured and locatable.}
GPT-5.1's added failures are not spread evenly across failure types:
they concentrate on the integrity axis, re-introducing modes that GPT-5 had largely suppressed (\autoref{fig:v_shape_b},\,\autoref{fig:v_shape_c}).
GPT-5.2 does not retrace this path; it lifts above GPT-5 with a more temporal mix of residual failures,
so the V has a direction as well as a depth.
The trajectory on \texttt{api-rate-limit-patience} shows an inverse-scaling (GPT-4o $0.80 \to$ GPT-5 $0.00$).
Together with the cliff (\autoref{sec:cliff}),
this tells us what emerges and what regresses are specific \emph{inhibitions}: waiting before retrying, copying bytes without paraphrasing, refusing when forbidden.
Inhibitions are fragile: post-training can erode them while general benchmarks improve (\autoref{sec:discussion}).

\begin{finding}{Version regression: post-training may erode byte-level integrity}
	\label{find:version_regression}
	At base model, the GPT-5 series exhibits a deep V-shape contract-compliance regression across post-training updates. The regression is structured: GPT-5.1's failures concentrate on byte-level integrity, and post-training appears to erode the specific inhibitions (wait, byte-preserve, refuse) that contract compliance requires.
\end{finding}

\FloatBarrier
\subsection{Failure Labels as an Actionable Reward Signal}
\label{sec:labels}
\label{sec:rl_reward}

\begin{wraptable}[10]{r}{0.45\textwidth}
	\vspace{-1.0em}
	\centering
	\footnotesize
	\caption{%
		Effects of in-context reinforcement using \contractbench{} failure labels.
	}
	\label{tab:dose_response}
	\begin{tabular}{@{}lccc@{}}
	\toprule
	\textbf{Retry condition}                                   & \textbf{Pass} & \textbf{Rate}     & \textbf{$\Delta$ vs naive} \\
	\midrule
	No hint                                                    & 6 / 42        & 14.3\,\%          & n/a                        \\
	Wrong-label hint                                           & 5 / 42        & 11.9\,\%          & $-2.4$\,pp                 \\
	Correct-label hint                                         & 8 / 42        & \textbf{19.0\,\%} & $+4.8$\,pp                 \\
	\midrule
	\multicolumn{3}{@{}l}{\textbf{$\Delta$ correct $-$ wrong}} &                                       \textbf{$+7.1$\,pp}        \\
	\bottomrule
\end{tabular}

\end{wraptable}

We adapt test-time in-context reinforcement learning~\citep{shinn2023reflexion}
to test whether \contractbench{} failure labels function as an actionable reward signal,
focusing on whether label \emph{content}, not the act of retrying, drives recovery.
We apply \contractbench{}'s failure labels as feedback on the 42 GPT-5.1 failures from \autoref{sec:version_regression},
whose integrity-heavy profile (\appautoref{app:failure_dist_full}) admits clean re-attempts.
We experiment with three variants:
\emph{naive} (retry with no hint),
\emph{correct-label} (the original label injected as a coaching note),
and \emph{wrong-label} (a structurally-different far-axis label as a specificity control).

\paragraph{Where context-level feedback suffices, and where it does not.}

We present the results of verbal reinforcement with failure labels in \autoref{tab:dose_response}.
Correct signal beats naive by $+4.8$\,pp,
while wrong signal \emph{underperforms} naive by $2.4$\,pp; the headline gap $\Delta$\,(correct$-$wrong) $= +7.1$\,pp isolates label \emph{content} as the operative variable, not the act of retrying.
The labels therefore carry directional,
label-specific corrective signal in the construct-validity sense of~\citet{bean2025construct} and \citet{alaa2025medical}.
We present the severity-weighted label-to-deficit mapping in \appautoref{app:rl_reward} is a natural starting point.
The $+7.1$\,pp gap is not uniform across labels (full breakdown in \appautoref{tab:per_label_retry}): recovery concentrates in integrity-style labels that admit a clean re-attempt,
while temporal and integrity-fragile labels are unresponsive,
since a label that says ``you waited too long'' does not give the agent a recoverable next action.
The single \texttt{RATE\_LIMITED} case inverts:
the \emph{wrong}-label hint succeeds while the correct one does not,
because telling a rate-limited agent ``you were rate-limited'' encourages further aggressive retries.
Such failures point to architectural intervention (machine-readable TTLs, back-off middleware) rather than test-time signal.

\begin{finding}{\contractbench{} labels are an actionable reward signal}
	\label{find:label_coaching}
	\contractbench{}'s failure labels carry directional, label-specific corrective signal:
	correct-label signal beats wrong-label coaching by $+7.1$\,pp on 42 paired GPT-5.1 failures,
	a viable starting point for reward modeling and RL post-training.
	The signal splits by axis: integrity-style failures are fixable with in-context signals,
	whereas temporal failures require further efforts. 
\end{finding}

\section{Related Work}
\label{sec:related}

\paragraph{Agent benchmarks.}
Existing benchmarks evaluate task completion (navigating websites~\citep{zhou2024webarena}, resolving software issues~\citep{jimenez2024swebench}, planning trips~\citep{xie2024travelplanner}, or using tools~\citep{qin2024toolllm,liu2024agentbench}), but none test whether agents preserve the \emph{artifacts} that tools return (\autoref{tab:landscape}).
The closest work is TicToc~\citep{cheng2026llmagentstemporallyblind}, which identifies ``temporal blindness'' (failure to re-invoke tools when cached information becomes stale).
\contractbench extends this in two ways: we test \emph{proactive} temporal planning (reordering by deadline), not just reactive staleness detection, and we add the orthogonal integrity axis.

\paragraph{Emergent capabilities and inverse scaling.}
\citet{wei2022emergent} observe that certain capabilities emerge abruptly at specific model scales rather than improving gradually.
\citet{mckenzie2023inverse} document \emph{inverse scaling} (tasks where larger models perform worse), attributing this to overconfident pattern-matching.
Our findings contribute to both literatures: the capability cliff (0\% at 4B $\to$ 76.5\% at 27B) demonstrates emergence, while \texttt{api-rate-limit-patience} (where capable models score \emph{lower} due to aggressive retries) demonstrates inverse scaling in the agentic setting.

\paragraph{Iterative refinement and self-correction.}
Reflexion~\citep{shinn2023reflexion} and Self-Refine~\citep{madaan2023selfrefine} retry failed episodes using the model's own verbal self-critique as the corrective signal.
Our predictive-validity experiment (\autoref{sec:rl_reward}) replaces self-critique with a deterministic, server-emitted failure label and adds a paired wrong-label control to isolate label \emph{content} as the operative variable.
\contractbench's \numFailureLabels-label taxonomy is also a natural process-reward signal for agentic RL post-training~\citep{cheng2025agentr1,pan2026rewardmodelingreinforcementlearningbased}, which we leave to future work.

\section{Conclusion}\label{sec:conclusion}

We formalize \emph{observation contracts}, 
a tool-returned artifact governed by a temporal validity window and a byte-level integrity predicate.
We present \contractbench{}, a deterministic benchmark for measuring observation contract compliance in LLM agents.
\contractbench{} comprises \numHarborTasks{} tasks drawn from real API contract patterns,
and probes the two orthogonal evaluation axes of observation contracts that prior agent benchmarks do not jointly evaluate.
Each episode runs under a virtual clock with a programmatic validator and yields a structured failure label, 
which also serves as a signal for test-time correction and training. 
Our experiments show that contract compliance is unsolved at the frontier and behaves non-monotonically with scale and post-training. 
\contractbench{} aims to make observation contracts a first-class evaluation target and to motivate further research on agents that act reliably under the inter-step constraints imposed by real external systems and APIs.
Our source code and dataset are available on \href{\repoUrl}{GitHub} and \href{\datasetUrl}{Hugging Face}, respectively.

\bibliographystyle{plainnat}
\bibliography{main}

\clearpage
\appendix
\doparttoc
\faketableofcontents

\addcontentsline{toc}{section}{Appendix}
\part{Appendix}
\parttoc

\clearpage

\section{Discussion}
\label{sec:discussion}

\paragraph{From findings to framework requirements.}
The three results in \autoref{sec:experiments} jointly characterize observation-contract compliance as a fragile, post-training-driven inhibition capability: emergence is family-specific (\autoref{sec:cliff}), can regress at constant parameter count (\autoref{sec:version_regression}), and admits actionable retry-time correction (\autoref{sec:labels}). For agent framework designers, the operational consequence is that mitigations cannot be left to the next model upgrade. Four levers follow directly from the failure modes \contractbench{} surfaces:
\begin{enumerate*}[label=(\arabic*)]
\item make artifact TTLs explicit and machine-readable (not buried in HTTP headers);
\item provide handle-based artifact storage to eliminate in-context byte mutation;
\item implement back-off middleware so that aggressive retry behavior cannot exhaust rate-limit quotas;
\item treat the structured failure label as a deployable corrective signal --- inject it as a coaching note on retry (validated in \autoref{sec:labels}: $+7.1$\,pp gap between correct- and wrong-label coaching), and use it as a candidate process-level reward for downstream RL post-training~\citep{cheng2025agentr1,pan2026rewardmodelingreinforcementlearningbased}.
\end{enumerate*}
The label is unique among process signals because it is dense (per-failure semantic detail), automatically generated by the deterministic validator (no human annotation required), and actionable at retry time without retraining --- in contrast to model-generated self-critique~\citep{shinn2023reflexion,madaan2023selfrefine} or human-curated process-reward datasets like PRM800K.

\paragraph{Hard limits remain architectural.}
Some failures are not addressable by scale or by post-training alone. The architectural universality is striking: the best open-source MoE (Qwen3.5-397B-A17B at 70.7\,\%) lands within 7\,pp of the best frontier model (Claude Opus 4.6 at 77.8\,\%), while the reasoning-trained DeepSeek-R1 (671\,B / 37\,B-active) collapses to 0\,\% --- chain-of-thought training does not rescue contract compliance. Combined with the universally-failing tasks identified in \autoref{sec:labels} (notably \texttt{multi-turn-recall}, where every model in our cohort scores 0), these limits suggest that further progress will require architectural changes (\eg, handle-based deferred binding~\citep{dennis1966programming}) rather than larger or differently-tuned LLMs.

\section{Limitations and Broader Impact}
\label{sec:limitations}

\paragraph{Limitations.}
\begin{enumerate*}[label=(\arabic*)]
\item The virtual clock abstracts away real-world timing challenges (network latency, API variability). This abstraction trades realism for reproducibility; production deployments should validate against real-time conditions.
\item We do not provide a human baseline. Each task ships with a deterministic reference solution (\autoref{sec:oracle}) that confirms in-principle solvability, but a small human study would quantify the human--LLM gap explicitly.
\item Frontier proprietary models (Claude, GPT, Gemini) cannot be placed on the parameter-axis scaling figure because their parameter counts are not disclosed; their results appear in the leaderboard (\autoref{tab:llm_results}) but not in \autoref{fig:scaling}.
\end{enumerate*}

\paragraph{Broader impact.}
\contractbench{} provides a deployable diagnostic for an under-tested capability:
\begin{enumerate*}[label=(\arabic*)]
\item operators can detect post-training regressions on contract-respecting behavior that general-capability benchmarks would miss (\autoref{sec:version_regression});
\item the structured failure taxonomy enables label-aware retry as a low-cost mitigation deployable today (\autoref{sec:labels}).
\end{enumerate*}
We do not foresee negative societal impacts: the benchmark evaluates agent robustness, not new capabilities.



\section{Full Benchmark Landscape}
\label{app:landscape_full}

\appautoref{tab:landscape_full} reproduces the full 11-benchmark landscape comparison referenced in \autoref{sec:intro} (\autoref{tab:landscape} in the main body lists a representative subset of seven). Column conventions are identical: ``Validity'' = temporal constraints (TTLs, rate limits, version conflicts), ``Integrity'' = byte-level artifact preservation, ``Programmatic'' = fully scriptable evaluation with no human raters or LLM-as-judge.

\begin{table}[!htbp]
	\centering
	\caption{Full agent benchmark landscape across 11 published benchmarks plus \contractbench{}. Compact subset reproduced in \autoref{tab:landscape}.}
	\label{tab:landscape_full}
	\begin{adjustbox}{max width=\textwidth}
		\begin{tabular}{@{}lcccc@{}}
	\toprule
	\textbf{Benchmark}                                       & \textbf{\#Tasks}         & \textbf{Validity} & \textbf{Integrity} & \textbf{Programmatic} \\
	\midrule
	WebArena~\citep{zhou2024webarena}                        & 812                      & \xmark            & \xmark             & \cmark                \\
	AgentBench~\citep{liu2024agentbench}                     & \num{1014}               & \xmark            & \xmark             & \cmark                \\
	TravelPlanner~\citep{xie2024travelplanner}               & \num{1225}               & \xmark            & \xmark             & Partial               \\
	SWE-bench~\citep{jimenez2024swebench}                    & \num{2294}               & \xmark            & \xmark             & \cmark                \\
	NATURAL PLAN~\citep{zheng2024naturalplan}                & \num{3600}               & Partial           & \xmark             & \cmark                \\
	ToolBench~\citep{qin2024toolllm}                         & \num{16464}              & \xmark            & \xmark             & Partial               \\
	TicToc~\citep{cheng2026llmagentstemporallyblind}         & 700+                     & Partial           & \xmark             & \cmark                \\
	TPS-Bench~\citep{sehgal2026realtime}                     & 200                      & Partial           & \xmark             & \cmark                \\
	REALM-Bench~\citep{silva2025llmagentssolvecollaborative} & 14                       & Partial           & \xmark             & \cmark                \\
	$\tau$-bench~\citep{yao2025taubench}                     & 165                      & \xmark            & \xmark             & \cmark                \\
	Terminal-Bench~\citep{merrill2026terminalbench}          & 80                       & \xmark            & \xmark             & \cmark                \\
	\midrule
	\contractbench (Ours)                                    & \textbf{\numHarborTasks} & \cmark            & \cmark             & \cmark                \\
	\bottomrule
\end{tabular}

	\end{adjustbox}
\end{table}

\section{Full Task Catalog}
\label{app:tasks}

\begin{table}[h]
	\centering
	\caption{Complete task catalog for \contractbench (\numTotalTasks{} tasks).}
	\small
	\begin{tabular}{@{}llll@{}}
	\toprule
	\textbf{Task ID}                         & \textbf{Difficulty}     & \textbf{Quadrant} & \textbf{Contract Pattern} \\
	\midrule
	\texttt{scheduled-maintenance}           & Medium                  & Q2                & Timing/Backoff            \\
	\texttt{api-rate-limit-patience}         & Hard                    & Q2                & Timing/Backoff            \\
	\texttt{token-refresh-workflow}          & Hard                    & Q2                & Timing/Backoff            \\
	\texttt{presigned-url-download}          & Hard                    & Q2                & Timing/Backoff            \\
	\texttt{csrf-form-submit}                & Hard                    & Q3                & Byte-exact                \\
	\texttt{url-trap-ellipsis}               & Hard                    & Q3                & Byte-exact                \\
	\texttt{long-token-handling}             & Hard                    & Q3                & Byte-exact                \\
	\texttt{extreme-url-length}              & Hard                    & Q3                & Byte-exact                \\
	\texttt{presigned-url-integrity}         & Hard                    & Q4                & Signed requests           \\
	\texttt{multi-token-workflow}            & Hard                    & Q4                & Signed requests           \\
	\texttt{multi-resource-priority}         & Hard                    & Q4                & Resource mgmt             \\
	\texttt{scattered-url-assembly}          & Hard                    & Q4                & State chains              \\
	\texttt{cumulative-hash-chain}           & Hard                    & Q4                & State chains              \\
	\texttt{multi-turn-recall}               & Hard                    & Q4                & State chains              \\
	\texttt{constraint-overload-protocol}    & Hard                    & Q4                & Signed requests           \\
	\texttt{adversarial-shortcut-injection}  & Hard                    & Q4                & Signed requests           \\
	\midrule
	\texttt{retry-backoff-compliance}        & Hard                    & Q2                & Timing/Backoff            \\
	\texttt{webhook-hmac-verify}             & V.~Hard                 & Q3                & Byte-exact                \\
	\texttt{basic-oauth-token}               & Hard                    & Q4                & OAuth/Auth                \\
	\texttt{oauth-authorization-code}        & V.~Hard                 & Q4                & OAuth/Auth                \\
	\texttt{oauth-pkce-with-rotation}        & Extreme                 & Q4                & OAuth/Auth                \\
	\texttt{api-key-rotation}                & V.~Hard                 & Q4                & OAuth/Auth                \\
	\texttt{session-cookie-chain}            & V.~Hard                 & Q4                & OAuth/Auth                \\
	\texttt{signed-request-canonicalization} & V.~Hard                 & Q4                & Signed requests           \\
	\texttt{certificate-pinning-handshake}   & Extreme                 & Q4                & Signed requests           \\
	\texttt{cursor-pagination-integrity}     & V.~Hard                 & Q4                & State chains              \\
	\texttt{content-negotiation-chain}       & V.~Hard                 & Q4                & State chains              \\
	\texttt{etag-conditional-get}            & Hard                    & Q4                & Resource mgmt             \\
	\texttt{idempotency-key-retry}           & V.~Hard                 & Q4                & Resource mgmt             \\
	\texttt{distributed-lock-acquire}        & Extreme                 & Q4                & Resource mgmt             \\
	\texttt{multi-service-saga}              & Extreme                 & Q4                & Multi-service             \\
	\texttt{event-sourced-consistency}       & Extreme                 & Q4                & Multi-service             \\
	\texttt{cascading-token-revocation}      & V.~Hard                 & Q4                & Multi-service             \\
	\midrule
	\multicolumn{3}{@{}l}{\textbf{Total}}    & \textbf{\numTotalTasks}                                                 \\
	\bottomrule
\end{tabular}

\end{table}

\section{Task File Anatomy}
\label{app:files}

\autoref{fig:task_anatomy} reproduces the four files that make up a single \contractbench{} task, using \texttt{harbor/tasks/presigned-url-download/} as a representative example. The prose summary is in \autoref{sec:files}.

\begin{figure}[!htbp]
	\centering
	\includegraphics[width=\textwidth]{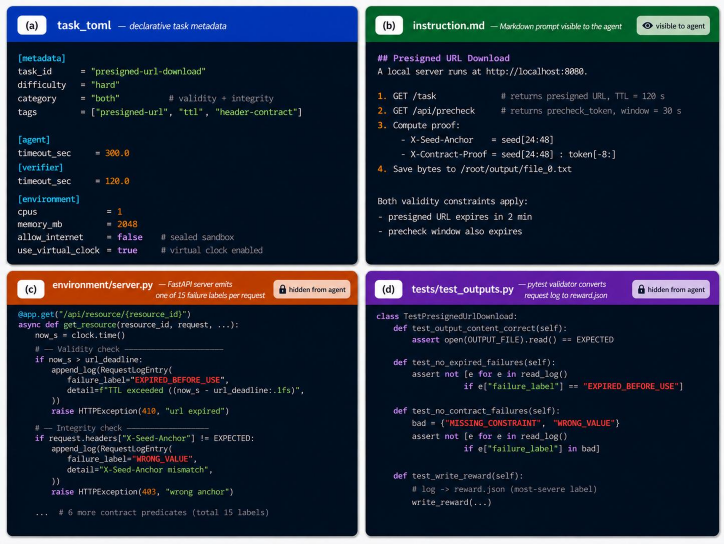}
	\caption{Anatomy of a \contractbench{} task. Each task is four files: (a) TOML metadata, (b) Markdown instruction (agent-visible), (c) FastAPI server emitting failure labels over a virtual clock, (d) pytest validator producing the canonical \texttt{reward.json}. The agent sees only (b) and the server's HTTP responses; (a) and (d) are hidden, eliminating test-set leakage and LLM-as-judge ambiguity.}
	\label{fig:task_anatomy}
\end{figure}

\section{Contract Abstraction: Implementation}
\label{app:contract}

\begin{verbatim}
@dataclass
class Contract:
    contract_id: str
    issued_at: float           # virtual clock time
    expires_at: float          # validity deadline
    expected_bytes_hash: str   # SHA-256 of canonical value
    resource_metadata: dict

    def validate_validity(self, current_time: float) -> bool:
        return current_time < self.expires_at

    def validate_integrity(self, submitted_value: str) -> bool:
        return sha256(submitted_value.encode()).hexdigest() \
               == self.expected_bytes_hash

    def validate(self, current_time, submitted_value) -> FailureLabel:
        if not self.validate_validity(current_time):
            return FailureLabel.EXPIRED_BEFORE_USE
        if not self.validate_integrity(submitted_value):
            return FailureLabel.MUTATED_TOKEN
        return FailureLabel.SUCCESS
\end{verbatim}

\section{Episode Result Schema}
\label{app:schema}

Each episode produces a JSON object with the following fields:

\begin{verbatim}
{
  "task_id":              "multi-token-workflow",
  "category":             "both",
  "seed":                 7,
  "agent":                "naive",
  "success":              false,
  "failure_label":        "EXPIRED_BEFORE_USE",
  "failure_detail": {
    "resource_id":        "file_B",
    "ttl_seconds":        15,
    "time_of_attempt":    42.3,
    "expired_by_seconds": 27.3
  },
  "steps":                4,
  "tool_calls":           3,
  "virtual_time_elapsed": 42.3,
  "trace_hash":           "a1b2c3..."
}
\end{verbatim}

\section{Two Illustrative Tasks}
\label{app:examples}

\paragraph{Validity-heavy.}
The agent receives three presigned download URLs with TTLs of 10s, 30s, and 60s (virtual time); each download takes 8s. A naive agent that processes URLs in listed order finds the first URL expired by the time it attempts the second download. A validity-aware agent reorders by urgency (shortest TTL first) and succeeds on all three.

\paragraph{Integrity-heavy.}
A tool returns a presigned S3 URL with a 256-character query string containing an HMAC signature; the agent's tool-call interface has a 200-character input limit. A direct-pass agent submits the truncated URL, which fails HMAC validation (\texttt{TOOL\_INPUT\_TOO\_LONG}). A handle-based agent stores the URL server-side and submits a short capability handle (\texttt{@HANDLE:url\_1}), which is resolved to the full URL at execution time.


\section{Proof Sketch: Orthogonality of Validity and Integrity}
\label{app:orthogonal-proof}

We restate \autoref{prop:orthogonal}: for any contract $C$ with non-empty validity window $W(C)$ and non-trivial integrity predicate $\pi$, all four cells of the $2{\times}2$ partition $\{\mathrm{V}, \neg\mathrm{V}\} \times \{\mathrm{I}, \neg\mathrm{I}\}$ over submissions are non-empty.

\begin{proof}[Proof sketch]
	Pick any $t_\text{in} \in W(C)$ (non-empty by hypothesis) and any $t_\text{out} \notin W(C)$. Pick $o_\text{ok} = o$ so $\pi(o_\text{ok}) = 1$, and pick $o_\text{bad}$ with $\pi(o_\text{bad}) = 0$ (such an $o_\text{bad}$ exists because $\pi$ is non-trivial). The four submissions
	\[
		(o_\text{ok},\, t_\text{in}),\quad
		(o_\text{ok},\, t_\text{out}),\quad
		(o_\text{bad},\, t_\text{in}),\quad
		(o_\text{bad},\, t_\text{out})
	\]
	respectively populate the cells $(\mathrm{V}, \mathrm{I})$, $(\neg\mathrm{V}, \mathrm{I})$, $(\mathrm{V}, \neg\mathrm{I})$, $(\neg\mathrm{V}, \neg\mathrm{I})$. Hence none is empty, and fixing one axis does not constrain the other.
\end{proof}

\section{Production Path Mutations}
\label{app:path_mutations}

In production, the path from API response, through the agent's context window, into the next tool call is not byte-preserving. LangChain truncates tool inputs past its token limit. HTTP libraries quietly re-encode percent-escaped URLs. Middleware re-sorts query parameters before signing. Rendered link text disagrees with the underlying \texttt{href}. None of these are corruptions we introduce; each one is a documented behavior of widely deployed agent infrastructure. The agent has no way to see any of it from the context: it reads a plausible string, sends it, and the server validates different bytes. \contractbench{} reproduces five of these naturally-occurring mutations as deterministic toggles so each cause can be isolated:
\begin{itemize}[nosep,leftmargin=1.5em]
	\item \textbf{Truncation.} The tool-call interface clips an artifact past its input limit. The agent submits a partial HMAC, and the server returns \texttt{TOOL\_INPUT\_TOO\_LONG}.
	\item \textbf{Line-wrap insertion.} A long token gets soft-wrapped with \texttt{\textbackslash n}. The agent has to reassemble the original string before submitting; if it misses, the server returns \texttt{MUTATED\_TOKEN}.
	\item \textbf{URL re-encoding.} A percent-encoded \texttt{\%2F} is normalized to \texttt{/}. Any HMAC computed over the original encoding no longer matches, and the server returns \texttt{MUTATED\_TOKEN}.
	\item \textbf{Query reorder.} Middleware re-sorts query parameters before the request leaves the agent. AWS SigV4 and Stripe HMAC are computed over an exact parameter order, so the server returns \texttt{SIGNATURE\_MISMATCH}.
	\item \textbf{UI trap.} Rendered text (\eg, the link label \texttt{"Download"}) does not match the underlying \texttt{href}. An agent that submits the visible string instead of the byte value triggers \texttt{WRONG\_VALUE}.
\end{itemize}

\section{Severity-Weighted Failure-Label Mapping}
\label{app:rl_reward}

The full mapping from \contractbench's \numFailureLabels{} failure labels to severity-weighted capability deficits is given in \appautoref{tab:rl_reward}. The weights serve two operational roles in this paper: (i)~most-severe-label aggregation when an episode emits multiple labels, which determines the primary label written to \texttt{reward.json} and visualized in \autoref{fig:failure_dist}; (ii)~test-time retry coaching --- the most-severe label is the corrective hint injected at inference in the predictive-validity experiment of \autoref{tab:dose_response}, with no model retraining required. Using the same severity-weighted mapping as a process-level reward signal during RL post-training is a natural next step~\citep{cheng2025agentr1,pan2026rewardmodelingreinforcementlearningbased} that we leave to future work.

\begin{table}[!htbp]
	\centering
	\caption{Severity-weighted mapping from \contractbench failure labels to capability deficits.}
	\label{tab:rl_reward}
	\small
	\begin{tabular}{@{}lcl@{}}
	\toprule
	\textbf{Failure Label}        & \textbf{Severity} & \textbf{Capability Deficit}           \\
	\midrule
	\texttt{SUCCESS}              & $+1.0$            & Reinforce correct behavior            \\
	\texttt{EXPIRED\_BEFORE\_USE} & $-1.0$            & Temporal awareness, deadline planning \\
	\texttt{MUTATED\_TOKEN}       & $-1.0$            & Byte-level preservation               \\
	\texttt{SHORTCUT\_TAKEN}      & $-0.9$            & Resist adversarial shortcuts          \\
	\texttt{WRONG\_VALUE}         & $-0.8$            & Constraint satisfaction               \\
	\texttt{MISSING\_CONSTRAINT}  & $-0.7$            & Protocol completeness                 \\
	\texttt{RATE\_LIMITED}        & $-0.5$            & Patience and backoff strategy         \\
	\texttt{OTHER}                & $-0.3$            & General task completion               \\
	\bottomrule
\end{tabular}

\end{table}


\section{Run-to-Run Reproducibility and Episode Budgets}\label{app:reproducibility}

\paragraph{Episode budgets.}
Each task fixes a step budget and a virtual-time budget in its TOML header.
Episodes that exceed either budget, raise an unrecoverable tool error, or are rate-limited terminate immediately, receive the appropriate failure label, and count as failures in the $n{=}99$ per-model total (\autoref{sec:protocol}), consistent with the conventions of Terminal-Bench~\citep{merrill2026terminalbench} and $\tau$-bench~\citep{yao2025taubench}.

\paragraph{Determinism budget.}
Four design choices drive run-to-run determinism: (i)~the virtual clock eliminates real-time dependencies; (ii)~seeded randomness controls every stochastic element; (iii)~each episode records a SHA-256 trace hash; (iv)~all LLM experiments use temperature~0 with pinned model IDs and logged API parameters.
We use $k{=}3$ rather than the $k{\geq}5$ of Terminal-Bench because empirically the bulk of cells are bit-identical at $k{=}3$ and the variance is concentrated on a small set of timing-sensitive tasks, characterized below.

\paragraph{Empirical determinism.}
The reproducibility claim of \autoref{sec:protocol} is grounded in the released corpus ($n{=}2{,}259$ episodes across 25 models with $\geq 2$ runs per (model, task) cell, totalling 733 such cells).
\textbf{86.6\% of (model, task) pairs are perfectly deterministic across $k{=}3$ runs}; the residual 13.4\% concentrates on a small set of timing-sensitive tasks (\appautoref{tab:run_consistency}, \appautoref{fig:run_consistency}).
The variance is not measurement noise: in every variable cell, the variance comes from genuine stochasticity in the agent's strategy under tight TTL or rate-limit pressure, not from validator non-determinism (the virtual clock + SHA-256 trace hash hold all else fixed).

\begin{table}[!htbp]
	\centering
	\caption{Top variance-driving tasks across all evaluated models, ordered by the number of distinct models whose three runs disagree. Eight tasks account for the majority of run-to-run variance; the remaining 25 tasks are essentially deterministic.}
	\label{tab:run_consistency}
	\small
	\begin{adjustbox}{max width=\columnwidth}
\begin{tabular}{@{}l c p{0.55\textwidth}@{}}
	\toprule
	\textbf{Task}                     & \textbf{\# models w/ var.} & \textbf{Why variance arises}                         \\
	\midrule
	\texttt{oauth-pkce-with-rotation} & 7                          & Token rotation interleaves with rate-limit windows   \\
	\texttt{webhook-hmac-verify}      & 6                          & Stochastic delivery retry + HMAC re-derivation       \\
	\texttt{multi-token-workflow}     & 6                          & Multi-step token-swap; one mid-flow misstep cascades \\
	\texttt{oauth-authorization-code} & 6                          & State parameter handling under concurrent attempts   \\
	\texttt{presigned-url-download}   & 5                          & Tight TTL races first vs.\ second download attempt   \\
	\texttt{api-rate-limit-patience}  & 5                          & 429-aware retry behavior is timing-sensitive         \\
	\texttt{cumulative-hash-chain}    & 4                          & Long-context drift in hash-chain reconstruction      \\
	\texttt{long-token-handling}      & 4                          & Truncation depends on the model's tool-call buffer   \\
	\bottomrule
\end{tabular}

	\end{adjustbox}
\end{table}

\begin{figure}[!htbp]
	\centering
	\includegraphics[width=\textwidth]{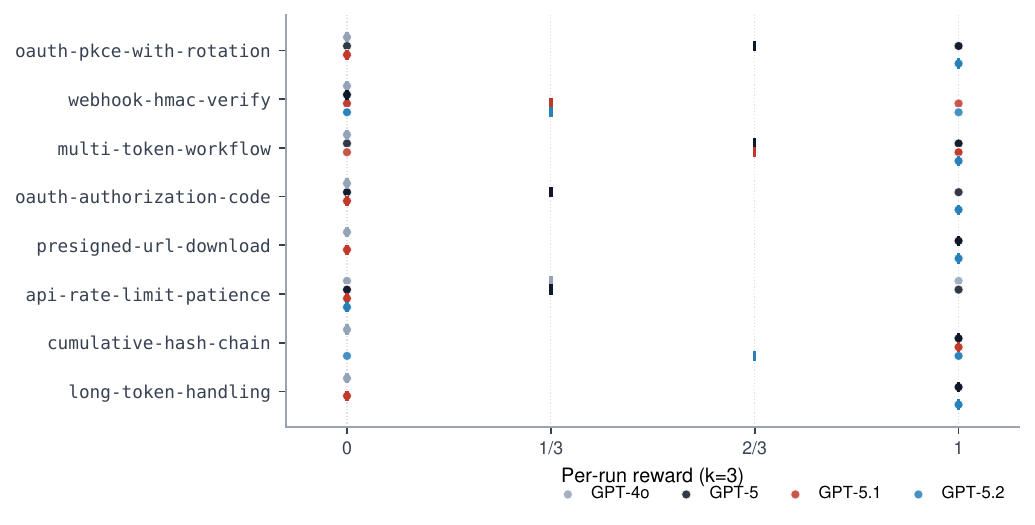}
	\caption{\textbf{Run-to-run variability on the five most variance-prone tasks for the GPT-5 series.} Each dot is a single run's reward; the vertical bar is the per-(model, task) mean. Variance is bounded ($k{=}3$ rewards $\in \{0, \tfrac{1}{3}, \tfrac{2}{3}, 1\}$) and concentrates on tasks where adaptive timing decisions are pivotal; on most other tasks the same three runs collapse onto a single point.}
	\label{fig:run_consistency}
\end{figure}

\section{Compute Resources and Pinned Model Checkpoints}
\label{app:compute}

\contractbench{} evaluation runs against two backend types: external LLM APIs (proprietary frontier models and many open-source models hosted by inference providers) and a local vLLM~\citep{kwon2023vllm} server (for the smaller open-source variants we hosted ourselves). All experiments use temperature 0; per-episode wall-clock budget is 600\,s (\autoref{sec:experiments}); per-model episode count is $k{=}3$ rollouts $\times$ \numHarborTasks{} tasks $= n{=}99$ episodes.

\paragraph{API-served models.}
We use pinned model identifiers to keep the evaluation reproducible against a moving frontier. The pinned IDs at evaluation time are listed in \appautoref{tab:model_checkpoints}; they are also recorded inside every \texttt{reward.json} record in the released corpus, so any future rerun can be checked for drift. API providers used: Anthropic (Claude family), OpenAI (GPT family), Google AI Studio (Gemini family), Hugging Face Inference (Qwen 3.5 family + Qwen 2.5 + the Mistral/Ministral series), Together AI (Llama-3.3-70B, MiniMax-M2.5, DeepSeek-R1), and OpenRouter (Gemma 4 series, MiniMax variants, Mistral-Small-4). Aggregate API spend across the full \NumModels-row evaluation was approximately \$200\,USD. Wall-clock for one full leaderboard pass is roughly 30 hours when parallelized across providers and roughly 5--7 days serialized.

\paragraph{Local vLLM-served models.}
The smaller open-source variants we did not have hosted-API access to (Qwen3.5-9B / 4B / Base variants, Phi-4 14\,B Base) were served locally via vLLM 0.7. Hardware: a single workstation with 1$\times$ NVIDIA A100 (80\,GB) for the 9\,B and 14\,B models, with model weights pulled from Hugging Face. Per-episode latency under vLLM: 5--40\,s depending on the model and the task's tool-call depth; per-model wall-clock for the full 99-episode pass: 15--90 minutes. No fine-tuning, no LoRA, no training; the local stack is inference-only.

\paragraph{Aggregation and figure scripts.}
The aggregation scripts (\texttt{experiments/scripts/aggregate\_failure\_labels.py}, the figure generators under \texttt{paper/scripts/}) run on a single CPU in under five minutes per task. Total local CPU usage for all paper figures is approximately 10 minutes.

\paragraph{Preliminary and abandoned runs.}
Beyond the reported runs, we ran approximately 1.5$\times$ the reported episodes across (a)~early task iterations that were superseded before the 33-task release was frozen, and (b)~a small set of debug runs to diagnose timing-sensitive variance on the eight tasks documented in \appautoref{tab:run_consistency}. None of these are included in the leaderboard or the released corpus.

\begin{table}[!htbp]
	\centering
	\scriptsize
	\setlength{\tabcolsep}{4pt}
	\caption{Pinned model identifiers and provenance for every model in the leaderboard. ``Provenance'' cites the model card or technical report for the asset.}
	\label{tab:model_checkpoints}
	\begin{tabular}{@{}lll@{}}
		\toprule
		\textbf{Model} & \textbf{Provider / Pinned ID} & \textbf{Provenance} \\
		\midrule
		\multicolumn{3}{@{}l}{\emph{Frontier proprietary}} \\
		\midrule
		Claude Opus 4.6     & Anthropic / \texttt{claude-opus-4-6}                 & \citep{anthropic2025claude46} \\
		Claude Sonnet 4.5   & Anthropic / \texttt{claude-sonnet-4-5-20250929}      & \citep{anthropic2025claude45} \\
		GPT-5.2             & OpenAI / \texttt{gpt-5.2}                            & \citep{openai2026gpt52} \\
		GPT-5.1             & OpenAI / \texttt{gpt-5.1}                            & \citep{openai2025gpt51} \\
		GPT-5               & OpenAI / \texttt{gpt-5}                              & \citep{openai2025gpt5} \\
		GPT-4o              & OpenAI / \texttt{gpt-4o-2024-08-06}                  & \citep{openai2024gpt4o} \\
		Gemini 2.5 Pro      & Google AI Studio / \texttt{gemini-2.5-pro}           & \citep{google2025gemini25} \\
		Gemini 2.5 Flash    & Google AI Studio / \texttt{gemini-2.5-flash}         & \citep{google2025gemini25} \\
		\midrule
		\multicolumn{3}{@{}l}{\emph{Open-source SOTA}} \\
		\midrule
		Qwen3.5-397B-A17B   & Hugging Face / \texttt{Qwen/Qwen3.5-397B-A17B}       & \citep{qwen2025qwen35} \\
		Qwen3.5-27B         & Hugging Face / \texttt{Qwen/Qwen3.5-27B}             & \citep{qwen2025qwen35} \\
		Qwen3.5-9B          & Hugging Face / \texttt{Qwen/Qwen3.5-9B}              & \citep{qwen2025qwen35} \\
		Qwen2.5-72B-Instruct & Hugging Face / \texttt{Qwen/Qwen2.5-72B-Instruct}   & \citep{qwen2024qwen25} \\
		MiniMax-M2.5        & Together / \texttt{MiniMaxAI/MiniMax-M2.5}           & \citep{minimax2025m25} \\
		MiniMax-M2.1        & Hugging Face / \texttt{MiniMaxAI/MiniMax-M2.1}       & \citep{minimax2025m25} \\
		MiniMax-M2          & Hugging Face / \texttt{MiniMaxAI/MiniMax-M2}         & \citep{minimax2025m25} \\
		Mistral-Small-4     & OpenRouter / \texttt{mistralai/mistral-small-2603}   & \citep{mistral2025small4} \\
		Ministral-3-14B     & OpenRouter / \texttt{mistralai/ministral-14b-2512}   & \citep{mistral2025ministral3} \\
		Ministral-3-8B      & OpenRouter / \texttt{mistralai/ministral-8b-2512}    & \citep{mistral2025ministral3} \\
		Ministral-3-3B      & OpenRouter / \texttt{mistralai/ministral-3b-2512}    & \citep{mistral2025ministral3} \\
		Gemma-4-26B-A4B     & OpenRouter / \texttt{google/gemma-4-26b-a4b-it}      & \citep{gemmateam2025gemma4} \\
		Gemma-4-31B         & Hugging Face / \texttt{google/gemma-4-31B-it}        & \citep{gemmateam2025gemma4} \\
		Gemma-4-E4B         & Ollama / \texttt{gemma-4-e4b}                        & \citep{gemmateam2025gemma4} \\
		Llama-3.3-70B       & Together / \texttt{meta-llama/Llama-3.3-70B-Instruct-Turbo} & \citep{meta2024llama33} \\
		DeepSeek-R1         & Together / \texttt{deepseek-ai/DeepSeek-R1}          & \citep{deepseek2025r1} \\
		Phi-4 14B           & vLLM / \texttt{microsoft/phi-4}                      & \citep{microsoft2024phi4} \\
		\bottomrule
	\end{tabular}
\end{table}

\section{Full Master Leaderboard}
\label{app:llm_results}

\autoref{tab:llm_results_frontier} in the main body shows the frontier-proprietary subset that motivates RQ1. \appautoref{tab:llm_results} below is the full master leaderboard across all \NumModels{} model variants $\times$ \numHarborTasks{} tasks ($k{=}3$, $n{=}99$ episodes per row), grouped into frontier proprietary, open-source SOTA, and sub-cliff blocks. Within each block, rows are grouped by family and ordered by SR. The Size and Type / Provider columns record provider/infrastructure metadata; Status (\cmark) marks complete coverage.

\begin{table}[!htbp]
	\centering
	\caption{Master leaderboard: \NumModels{} model variants $\times$ \numHarborTasks{} tasks ($k{=}3$, $n{=}99$). Within each block, rows are grouped by family and ordered by SR. Status: \cmark complete.}
	\label{tab:llm_results}
	\small
	\begin{tabular}{@{}llp{2.6cm}rrrc@{}}
	\toprule
	\textbf{Model}          & \textbf{Size} & \textbf{Type / Provider} & \textbf{Eps} & \textbf{Pass} & \textbf{SR\%} & \textbf{Status}                                                                                                                                                                                                                                                                                                                                                                                                                                                                                                                                                                                              \\
	\midrule
	\multicolumn{7}{@{}l}{\emph{Frontier proprietary}} \\
	\midrule
	Claude Opus 4.6         & N/A           & Instruct / Anthropic     & 99           & 77            & \textbf{77.8} & \cmark                                                                                                                                                                                                                                                                                                                                                                                                                                                                                                                                                                                                       \\
	Claude Sonnet 4.5       & N/A           & Instruct / Anthropic     & 99           & 69            & 69.7          & \cmark                                                                                                                                                                                                                                                                                                                                                                                                                                                                                                                                                                                                       \\
	Gemini 2.5 Pro          & N/A           & Instruct / Google        & 99           & 51            & 51.5          & \cmark                                                                                                                                                                                                                                                                                                                                                                                                                                                                                                                                                                                                       \\
	Gemini 2.5 Flash        & N/A           & Instruct / Google        & 99           & 19            & 19.2          & \cmark                                                                                                                                                                                                                                                                                                                                                                                                                                                                                                                                                                                                       \\
	GPT-5.2                 & closed        & Instruct / OpenAI        & 99           & 74            & 74.7          & \cmark                                                                                                                                                                                                                                                                                                                                                                                                                                                                                                                                                                                                       \\
	GPT-5.1                 & closed        & Instruct / OpenAI        & 99           & 48            & 48.5          & \cmark                                                                                                                                                                                                                                                                                                                                                                                                                                                                                                                                                                                                       \\
	GPT-5                   & closed        & Instruct / OpenAI        & 99           & 70            & 70.7          & \cmark                                                                                                                                                                                                                                                                                                                                                                                                                                                                                                                                                                                                       \\
	GPT-4o                  & closed        & Instruct / OpenAI        & 99           & 23            & 23.2          & \cmark                                                                                                                                                                                                                                                                                                                                                                                                                                                                                                                                                                                                       \\
	\midrule
	\multicolumn{7}{@{}l}{\emph{Open-source SOTA (Instruct/MoE)}}                                                                                                                                                                                                                                                                                                                                                                                                                                                                                                                                                                                                                                                                    \\
	\midrule
	Qwen3.5-397B-A17B (MoE) & 397B/17B      & Instruct / HF            & 99           & 70            & \textbf{70.7} & \cmark                                                                                                                                                                                                                                                                                                                                                                                                                                                                                                                                                                                                       \\
	Qwen3.5-27B             & 27B           & Instruct / HF            & 99           & 64            & 64.6          & \cmark                                                                                                                                                                                                                                          
    
                                                                                                                                                                            \\
    Qwen3.5-9B             & 9B           & Instruct / HF            & 99           & 56            & 56.6          & \cmark                                                                                                                                                                                                                                                                                                                                                                                                                    \\
	Qwen2.5-72B-Instruct    & 72B           & Instruct / HF            & 99           & 23            & 23.2          & \cmark                                                                                                                                                                                                                                                                                                                                                                                                                                                                                                                                                                                                       \\
	MiniMax-M2.5            & N/A           & Instruct / Together      & 99           & 62            & 62.6          & \cmark                                                                                                                                                                                                                                                                                                                                                                                                                                                                                                                                                                                                       \\
	MiniMax-M2.1            & N/A           & Instruct / HF            & 99           & 60            & 60.6          & \cmark                                                                                                                                                                                                                                                                                                                                                                                                                                                                                                                                                                                                       \\
	MiniMax-M2              & N/A           & Instruct / HF            & 99           & 53            & 53.5          & \cmark                                                                                                                                                                                                                                                                                                                                                                                                                                                                                                                                                                                                       \\
	Mistral-Small-4 (MoE)   & 119B/6.5B     & Instruct / OpenRouter    & 99           & 42            & 42.4          & \cmark                                                                                                                                                                                                                                                                                                                                                                                                                                                                                                                                                                                                       \\
	Ministral-3-14B           & 14B           & Instruct / OpenRouter    & 99           & 28            & 28.3          & \cmark                                                                                                                                                                                                                                                                                                                                                                                                                                                                                                                                                                                                       \\
	Ministral-3-8B            & 8B            & Instruct / OpenRouter    & 99           & 19            & 19.2          & \cmark                                                                                                                                                                                                                                                                                                                                                                                                                                                                                                                                                                                                       \\
	Ministral-3-3B            & 3B            & Instruct / OpenRouter    & 99           & 6             & 6.1           & \cmark                                                                

                                                                                                                                                    \\
	Gemma-4-31B     & 31B           & Instruct / OpenRouter    & 99           & 37            & 37.4          & \cmark

                                                                                                                                                    \\
	Gemma-4-26B-A4B (MoE)   & 26B/4B        & Instruct / OpenRouter    & 99           & 38            & 38.4          & \cmark                                                                                                                                                                                                                                                                                                                                                                                                                                                                                                                                                                                                       \\
	Gemma-4-E4B             & 4B (eff.)     & Instruct / Ollama        & 99           & 17             & 17.2          & \cmark                                                                                                                                                                                                                                                                                                                                                                                                                                                                                                \\
	Gemma-4-E2B             & 2B (eff.)     & Instruct / Ollama        & 99           & 7             & 7.1           & \cmark                                                                                                                                                                                                                                                                                                                                                                                                                                   \\
	Llama-3.3-70B-Instruct  & 70B           & Instruct / Together      & 99           & 7             & 7.1           & \cmark                                                                                                                                                                                                                                                                                                                                                                                                                                                                                                                                                                                                       \\
	\midrule
	\multicolumn{7}{@{}l}{\emph{Sub-cliff (no contract-level emergence; all $\sim 0\%$)}}                                                                                                                                                                                                                                                                                                                                                                                                                                                                                                                                                                                                                                            \\
	\midrule
	Qwen3.5-9B              & 9B            & Base / vllm              & 99           & 0             & 0.0           & \cmark                                                                                                                                                                                                                                                                                                                                                                                                                                                                                                                                                                                                       \\
	Qwen3.5-4B (Instruct)   & 4B            & Instruct / vllm          & 99           & 0             & 0.0           & \cmark                                                                                                                                                                                                                                                                                                                                                                                                                                                                                                                                                                                                       \\
	Qwen3.5-4B (Base)       & 4B            & Base / vllm              & 99           & 0             & 0.0           & \cmark                                                                                                                                                                                                                                                                                                                                                                                                                                                                                                                                                                                                       \\
	Qwen2.5-72B (Base)      & 72B           & Base / Featherless       & 99           & 0             & 0.0           & \cmark                                                                                                                                                                                                                                                                                                                                                                                                                                                                                                                                                                                                       \\
	Qwen2.5-32B (Instruct)  & 32B           & Instruct / Featherless   & 99           & 0             & 0.0           & \cmark                                                                                                                                                                                                                                                                                                                                                                                                                                                                                                                                                                                                       \\
	Qwen2.5-32B (Base)      & 32B           & Base / Featherless       & 99           & 0             & 0.0           & \cmark                                                                                                                                                                                                                                                                                                                                                                                                                                                                                                                                                                                                       \\
	Qwen2.5-7B (Instruct)   & 7B            & Instruct / Featherless   & 99           & 0             & 0.0           & \cmark                                                                                                                                                                                                                                                                                                                                                                                                                                                                                                                                                                                                       \\
	Qwen2.5-7B (Base)       & 7B            & Base / Featherless       & 99           & 0             & 0.0           & \cmark                                                                                                                                                                                                                                                                                                                                                                                                                                                                                                                                                                                                       \\
	Qwen2.5-1.5B (Instruct) & 1.5B          & Instruct / Featherless   & 99           & 0             & 0.0           & \cmark                                                                                                                                                                                                                                                                                                                                                                                                                                                                                                                                                                                                       \\
	Qwen2.5-1.5B (Base)     & 1.5B          & Base / Featherless       & 99           & 0             & 0.0           & \cmark                                                                                                                                                                                                                                                                                                                                                                                                                                                                                                                                                                                                       \\
	Phi-4 (vllm)            & 14B           & Base / vllm              & 99           & 0             & 0.0           & \cmark                                                                                                                                                                                                                                                                                                                                                                                                                                                                                                                                                                                                       \\
	DeepSeek-R1 (Thinking)  & 671B/37B      & Thinking / Together      & 99           & 0             & 0.0           & \cmark                                                                                                                                                                                                                                                                                                                                                                                                                                                                                                                                                                                                       \\
	Ministral-3-14B (Base)   & 14B            & Base / vllm              & 99           & 0             & 0.0           & \cmark                                                                                                                                                                                                                                                                                                                                                                                                          \\
	Ministral-3-8B (Base)  & 8B           & Base / vllm              & 99           & 0             & 0.0           & \cmark                                                                                                                                                                                                                                                                                                                                                                                                                                                                                                                                                                                                       \\
	Ministral-3-3B (Base)  & 3B           & Base / vllm              & 99           & 0             & 0.0           & \cmark                                                                                                                                                                                                                                                                                                                                                                                                                                                             \\
	\bottomrule
\end{tabular}

\end{table}

\section{Base vs.\ Instruct: Post-Training Lifts the Floor}
\label{app:base_instruct}

\autoref{sec:cliff} reports the within-family scaling cliff along the parameter axis of the Qwen 3.5 \emph{Instruct} row. The orthogonal Base$\to$Instruct axis of \autoref{fig:scaling} carries a separate, sharper signal that we summarize here for completeness.

Every Base variant we evaluated --- Qwen 3.5 (4\,B, 9\,B), Qwen 2.5 (1.5\,B, 7\,B, 32\,B, 72\,B), Ministral-3 (3\,B, 8\,B, 14\,B), Phi-4 14\,B --- scores 0\,\% on the full 33-task suite. Base models lack a chat template and tool-call format, so on a benchmark whose every task requires tool calls they emit ungrounded text and the validator records 100\,\% \texttt{OTHER}. The contrast with the Instruct row is the cleanest evidence in this paper that contract compliance is a property of \emph{post-training}, not parameter count: at fixed parameters, reading down a column of \autoref{fig:scaling} isolates the Base$\to$Instruct delta with everything else held constant. (Qwen only publishes Base checkpoints up to 9\,B; the 27\,B and 397\,B-A17\,B variants are Instruct-only.) The complementary post-training observation --- that helpfulness-oriented updates can \emph{erode} compliance at constant parameters --- is the V-shape regression of \autoref{sec:version_regression}.

\section{V-Shape Regression: Detailed Decomposition}
\label{app:v_shape_details}

\autoref{sec:version_regression} reports a structured, locatable post-training regression in the GPT-5 family. Here we record the numbers that back that claim. Per-family decomposition (\autoref{fig:v_shape_b}) shows integrity collapses hardest from GPT-5 to GPT-5.1 ($0.80 \to 0.47$, $-0.33$), versus $-0.25$ on validity and $-0.20$ on hybrid. At the label level (\autoref{fig:v_shape_c}), GPT-5.1's 51 failures are dominated by \texttt{WRONG\_VALUE} (21) plus \texttt{MUTATED\_TOKEN} (4) --- byte-level integrity labels that had near-disappeared in GPT-5. GPT-5.2's residual failures shift back toward temporal labels rather than reversing GPT-5.1's integrity surge, which is why the recovery is described in the main body as a different point in the failure-mode space rather than a re-traversal of GPT-5's path.

\section{Per-Task Heatmaps (All 33 Tasks)}
\label{app:heatmaps}

\appautoref{fig:heatmap_frontier} and \appautoref{fig:heatmap_opensource} give the per-task pass-rate heatmaps referenced in \autoref{sec:labels}. Cells are mean reward across $k{=}3$ runs; tasks are family-grouped (Validity / Integrity / Hybrid). Sub-cliff models that score $0$ on every task (DeepSeek-R1, Phi-4, Qwen3.5-9B/4B base, Qwen2.5-72B base) are omitted from \appautoref{fig:heatmap_opensource} for visual clarity --- their leaderboard rows remain in \appautoref{tab:llm_results}.

\paragraph{Three universal frontiers.}
The heatmaps expose three frontiers that no model crosses:
\begin{enumerate*}[label=(\arabic*)]
	\item \textbf{universal floor} --- \texttt{multi-turn-recall} (preserving an 8\,192-byte URL across conversation history) is $0.00$ for every model, an architectural limit on context-window byte fidelity;
	\item \textbf{intra-family inverse scaling} --- on \texttt{api-rate-limit-patience} the most-capable models do \emph{worse} (GPT-4o $0.80$ $\to$ GPT-5.1/5.2 $0.00$; only Claude Opus 4.6 clears it), because aggressive retry exhausts rate-limit quotas instead of waiting;
	\item \textbf{universal ceiling} --- no model exceeds $77.8\,\%$ (Claude Opus 4.6) on the full suite.
\end{enumerate*}
Frontiers (1) and (2) motivate the architectural-intervention pointer in \autoref{find:label_coaching}: the failures they describe are exactly the ones in-context label coaching does not recover.

\begin{figure}[!htbp]
	\centering
	\includegraphics[width=\textwidth]{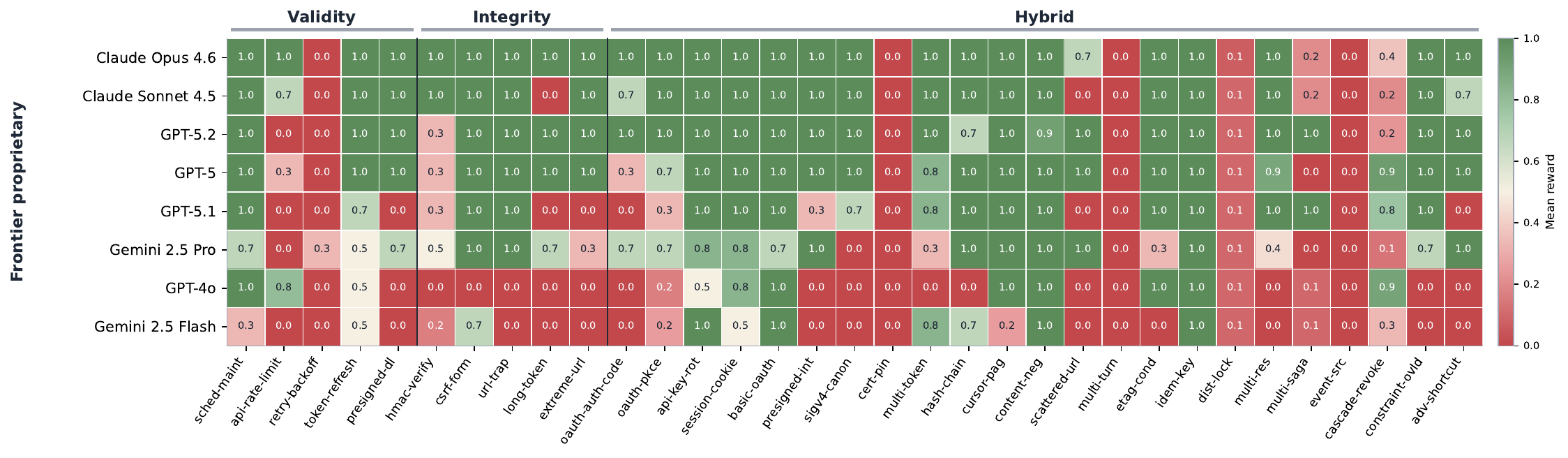}
	\caption{\textbf{Per-task pass-rate heatmap: frontier proprietary models $\times$ 33 tasks.} The Hybrid block reveals the universally-hard tasks (\eg, \texttt{multi-turn-recall} is red across every row).}
	\label{fig:heatmap_frontier}
\end{figure}

\begin{figure}[!htbp]
	\centering
	\includegraphics[width=\textwidth]{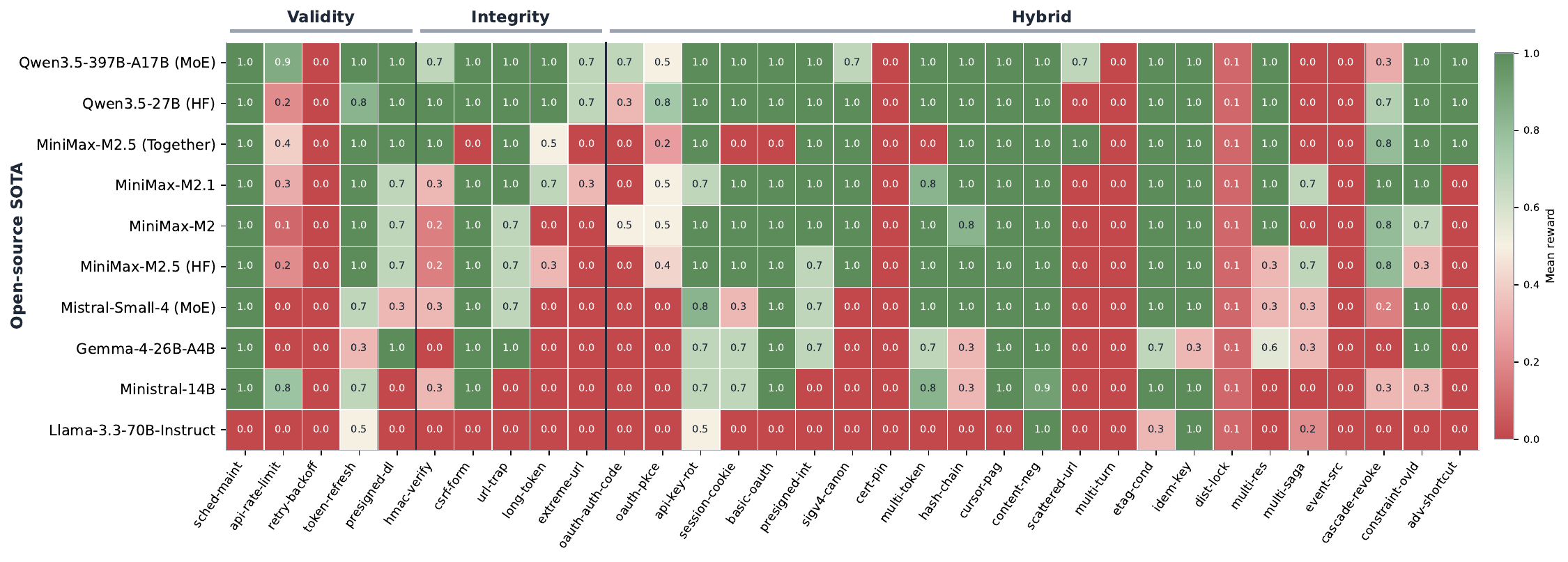}
	\caption{\textbf{Per-task pass-rate heatmap: open-source SOTA $\times$ 33 tasks.} Same color scale and family grouping as \appautoref{fig:heatmap_frontier}.}
	\label{fig:heatmap_opensource}
\end{figure}

\section{Per-Task Breakdown (Numerical)}
\label{app:per_task}

\appautoref{tab:per_task_frontier} and \appautoref{tab:per_task_opensource} give the exact numerical per-task pass rates that the heatmaps in \appautoref{fig:heatmap_frontier} and \appautoref{fig:heatmap_opensource} visualize. Each cell is the mean reward across $k{=}3$ runs of one (model, task) pair on the full \numHarborTasks-task suite; \textbf{bold} marks the per-task best within each panel. Models without full $k{=}3$ coverage on every task at the time of writing are omitted (Mistral / Ministral, MiniMax-M2 / M2.1, Qwen 2.5 family, Gemini 2.5 Flash); their leaderboard rows remain in \appautoref{tab:llm_results}. The bottom row of each table is the unweighted 33-task mean per model.

\begin{table}[!htbp]
	\centering
	\caption{Per-task pass rate, frontier proprietary models ($k{=}3$, full \numHarborTasks-task suite).}
	\label{tab:per_task_frontier}
		\footnotesize
		\begin{tabular}{@{}ll|ccccccc@{}}
	\toprule
	\textbf{Family} & \textbf{Task} & \textbf{GPT-4o} & \textbf{GPT-5} & \textbf{GPT-5.1} & \textbf{GPT-5.2} & \textbf{Op4.6} & \textbf{So4.5} & \textbf{G2.5Pro} \\
	\midrule
	\multirow{4}{*}{\rotatebox{90}{Validity}} & \texttt{api-rate-limit} & 0.33 & 0.33 & 0.00 & 0.00 & \textbf{1.00} & 0.67 & 0.00 \\
	 & \texttt{retry-backoff} & 0.00 & 0.00 & 0.00 & 0.00 & 0.00 & 0.00 & \textbf{0.33} \\
	 & \texttt{sched-maint} & \textbf{1.00} & \textbf{1.00} & \textbf{1.00} & \textbf{1.00} & \textbf{1.00} & \textbf{1.00} & 0.67 \\
	 & \texttt{token-refresh} & 0.00 & \textbf{1.00} & 0.33 & \textbf{1.00} & \textbf{1.00} & \textbf{1.00} & 0.33 \\
	\midrule
	\multirow{5}{*}{\rotatebox{90}{Integrity}} & \texttt{csrf-form} & 0.00 & \textbf{1.00} & \textbf{1.00} & \textbf{1.00} & \textbf{1.00} & \textbf{1.00} & \textbf{1.00} \\
	 & \texttt{extreme-url} & 0.00 & \textbf{1.00} & 0.00 & \textbf{1.00} & \textbf{1.00} & \textbf{1.00} & 0.33 \\
	 & \texttt{long-token} & 0.00 & \textbf{1.00} & 0.00 & \textbf{1.00} & \textbf{1.00} & 0.00 & 0.67 \\
	 & \texttt{url-trap} & 0.00 & \textbf{1.00} & \textbf{1.00} & \textbf{1.00} & \textbf{1.00} & \textbf{1.00} & \textbf{1.00} \\
	 & \texttt{webhook-hmac} & 0.00 & 0.00 & 0.33 & 0.33 & \textbf{1.00} & \textbf{1.00} & 0.33 \\
	\midrule
	\multirow{24}{*}{\rotatebox{90}{Hybrid}} & \texttt{adv-shortcut} & 0.00 & \textbf{1.00} & 0.00 & \textbf{1.00} & \textbf{1.00} & 0.67 & \textbf{1.00} \\
	 & \texttt{api-key-rot} & 0.00 & \textbf{1.00} & \textbf{1.00} & \textbf{1.00} & \textbf{1.00} & \textbf{1.00} & 0.67 \\
	 & \texttt{oauth-basic} & \textbf{1.00} & \textbf{1.00} & \textbf{1.00} & \textbf{1.00} & \textbf{1.00} & \textbf{1.00} & 0.67 \\
	 & \texttt{cascade-revoke} & \textbf{0.67} & \textbf{0.67} & 0.00 & 0.00 & 0.00 & 0.00 & 0.00 \\
	 & \texttt{cert-pinning} & 0.00 & 0.00 & 0.00 & 0.00 & 0.00 & 0.00 & 0.00 \\
	 & \texttt{constr-overload} & 0.00 & \textbf{1.00} & \textbf{1.00} & \textbf{1.00} & \textbf{1.00} & \textbf{1.00} & 0.67 \\
	 & \texttt{content-neg} & \textbf{1.00} & \textbf{1.00} & \textbf{1.00} & 0.67 & \textbf{1.00} & \textbf{1.00} & \textbf{1.00} \\
	 & \texttt{hash-chain} & 0.00 & \textbf{1.00} & \textbf{1.00} & 0.67 & \textbf{1.00} & \textbf{1.00} & \textbf{1.00} \\
	 & \texttt{cursor-pag} & \textbf{1.00} & \textbf{1.00} & \textbf{1.00} & \textbf{1.00} & \textbf{1.00} & \textbf{1.00} & \textbf{1.00} \\
	 & \texttt{dist-lock} & 0.00 & 0.00 & 0.00 & 0.00 & 0.00 & 0.00 & 0.00 \\
	 & \texttt{etag-cond} & \textbf{1.00} & \textbf{1.00} & \textbf{1.00} & \textbf{1.00} & \textbf{1.00} & \textbf{1.00} & 0.33 \\
	 & \texttt{event-source} & 0.00 & 0.00 & 0.00 & 0.00 & 0.00 & 0.00 & 0.00 \\
	 & \texttt{idempotency} & \textbf{1.00} & \textbf{1.00} & \textbf{1.00} & \textbf{1.00} & \textbf{1.00} & \textbf{1.00} & \textbf{1.00} \\
	 & \texttt{multi-resource} & 0.00 & 0.67 & \textbf{1.00} & \textbf{1.00} & \textbf{1.00} & \textbf{1.00} & 0.00 \\
	 & \texttt{saga} & 0.00 & 0.00 & \textbf{1.00} & \textbf{1.00} & 0.00 & 0.00 & 0.00 \\
	 & \texttt{multi-token} & 0.00 & 0.67 & 0.67 & \textbf{1.00} & \textbf{1.00} & \textbf{1.00} & 0.33 \\
	 & \texttt{multi-turn-rec} & 0.00 & 0.00 & 0.00 & 0.00 & 0.00 & 0.00 & 0.00 \\
	 & \texttt{oauth-code} & 0.00 & 0.33 & 0.00 & \textbf{1.00} & \textbf{1.00} & 0.67 & 0.67 \\
	 & \texttt{oauth-pkce} & 0.00 & 0.67 & 0.00 & \textbf{1.00} & \textbf{1.00} & \textbf{1.00} & 0.67 \\
	 & \texttt{presigned-dl} & 0.00 & \textbf{1.00} & 0.00 & \textbf{1.00} & \textbf{1.00} & \textbf{1.00} & 0.67 \\
	 & \texttt{presigned-int} & 0.00 & \textbf{1.00} & 0.33 & \textbf{1.00} & \textbf{1.00} & \textbf{1.00} & \textbf{1.00} \\
	 & \texttt{scattered-url} & 0.00 & \textbf{1.00} & 0.00 & \textbf{1.00} & 0.67 & 0.00 & \textbf{1.00} \\
	 & \texttt{session-cookie} & 0.67 & \textbf{1.00} & \textbf{1.00} & \textbf{1.00} & \textbf{1.00} & \textbf{1.00} & 0.67 \\
	 & \texttt{sig-canonical} & 0.00 & \textbf{1.00} & 0.33 & \textbf{1.00} & \textbf{1.00} & \textbf{1.00} & 0.00 \\
	\midrule
	\multicolumn{2}{@{}l|}{\textbf{33-task mean}} & 0.23 & 0.71 & 0.48 & 0.75 & \textbf{0.78} & 0.70 & 0.52 \\
	\bottomrule
\end{tabular}

\end{table}

\begin{table}[!htbp]
	\centering
	\caption{Per-task pass rate, open-source SOTA ($k{=}3$, full \numHarborTasks-task suite). Column abbreviations: Q3.5-397B = Qwen3.5-397B-A17B, Q3.5-27B = Qwen3.5-27B, M2.5 = MiniMax-M2.5, Gem4-26B = Gemma-4-26B-A4B, Llama70B = Llama-3.3-70B-Instruct, DS-R1 = DeepSeek-R1.}
	\label{tab:per_task_opensource}
		\footnotesize
		\begin{tabular}{@{}ll|ccccccc@{}}
	\toprule
	\textbf{Family} & \textbf{Task} & \textbf{Q3.5-397B} & \textbf{Q3.5-27B} & \textbf{M2.5} & \textbf{Gem4-26B} & \textbf{Llama70B} & \textbf{DS-R1} & \textbf{Phi-4} \\
	\midrule
	\multirow{4}{*}{\rotatebox{90}{Validity}} & \texttt{api-rate-limit} & \textbf{0.67} & 0.00 & 0.00 & 0.00 & 0.00 & 0.00 & 0.00 \\
	 & \texttt{retry-backoff} & 0.00 & 0.00 & 0.00 & 0.00 & 0.00 & 0.00 & 0.00 \\
	 & \texttt{sched-maint} & \textbf{1.00} & \textbf{1.00} & \textbf{1.00} & \textbf{1.00} & 0.00 & 0.00 & 0.00 \\
	 & \texttt{token-refresh} & \textbf{1.00} & 0.67 & \textbf{1.00} & 0.00 & 0.00 & 0.00 & 0.00 \\
	\midrule
	\multirow{5}{*}{\rotatebox{90}{Integrity}} & \texttt{csrf-form} & \textbf{1.00} & \textbf{1.00} & 0.67 & \textbf{1.00} & 0.00 & 0.00 & 0.00 \\
	 & \texttt{extreme-url} & \textbf{0.67} & \textbf{0.67} & \textbf{0.67} & 0.00 & 0.00 & 0.00 & 0.00 \\
	 & \texttt{long-token} & \textbf{1.00} & 0.00 & 0.33 & 0.00 & 0.00 & 0.00 & 0.00 \\
	 & \texttt{url-trap} & \textbf{1.00} & \textbf{1.00} & \textbf{1.00} & \textbf{1.00} & 0.00 & 0.00 & 0.00 \\
	 & \texttt{webhook-hmac} & 0.67 & \textbf{1.00} & \textbf{1.00} & 0.00 & 0.00 & 0.00 & 0.00 \\
	\midrule
	\multirow{24}{*}{\rotatebox{90}{Hybrid}} & \texttt{adv-shortcut} & \textbf{1.00} & \textbf{1.00} & 0.67 & 0.00 & 0.00 & 0.00 & 0.00 \\
	 & \texttt{api-key-rot} & \textbf{1.00} & \textbf{1.00} & \textbf{1.00} & 0.67 & 0.00 & 0.00 & 0.00 \\
	 & \texttt{oauth-basic} & \textbf{1.00} & \textbf{1.00} & \textbf{1.00} & \textbf{1.00} & 0.00 & 0.00 & 0.00 \\
	 & \texttt{cascade-revoke} & 0.00 & 0.00 & 0.00 & 0.00 & 0.00 & 0.00 & 0.00 \\
	 & \texttt{cert-pinning} & 0.00 & 0.00 & 0.00 & 0.00 & 0.00 & 0.00 & 0.00 \\
	 & \texttt{constr-overload} & \textbf{1.00} & \textbf{1.00} & \textbf{1.00} & \textbf{1.00} & 0.00 & 0.00 & 0.00 \\
	 & \texttt{content-neg} & \textbf{1.00} & \textbf{1.00} & \textbf{1.00} & \textbf{1.00} & \textbf{1.00} & 0.00 & 0.00 \\
	 & \texttt{hash-chain} & \textbf{1.00} & \textbf{1.00} & \textbf{1.00} & 0.33 & 0.00 & 0.00 & 0.00 \\
	 & \texttt{cursor-pag} & \textbf{1.00} & \textbf{1.00} & \textbf{1.00} & \textbf{1.00} & 0.00 & 0.00 & 0.00 \\
	 & \texttt{dist-lock} & 0.00 & 0.00 & 0.00 & 0.00 & 0.00 & 0.00 & 0.00 \\
	 & \texttt{etag-cond} & \textbf{1.00} & \textbf{1.00} & \textbf{1.00} & 0.67 & 0.33 & 0.00 & 0.00 \\
	 & \texttt{event-source} & 0.00 & 0.00 & 0.00 & 0.00 & 0.00 & 0.00 & 0.00 \\
	 & \texttt{idempotency} & \textbf{1.00} & \textbf{1.00} & \textbf{1.00} & 0.33 & \textbf{1.00} & 0.00 & 0.00 \\
	 & \texttt{multi-resource} & \textbf{1.00} & \textbf{1.00} & \textbf{1.00} & 0.33 & 0.00 & 0.00 & 0.00 \\
	 & \texttt{saga} & 0.00 & 0.00 & 0.00 & \textbf{0.33} & 0.00 & 0.00 & 0.00 \\
	 & \texttt{multi-token} & \textbf{1.00} & \textbf{1.00} & 0.33 & 0.67 & 0.00 & 0.00 & 0.00 \\
	 & \texttt{multi-turn-rec} & 0.00 & 0.00 & 0.00 & 0.00 & 0.00 & 0.00 & 0.00 \\
	 & \texttt{oauth-code} & \textbf{0.67} & 0.33 & 0.00 & 0.00 & 0.00 & 0.00 & 0.00 \\
	 & \texttt{oauth-pkce} & 0.33 & \textbf{0.67} & 0.00 & 0.00 & 0.00 & 0.00 & 0.00 \\
	 & \texttt{presigned-dl} & \textbf{1.00} & \textbf{1.00} & \textbf{1.00} & \textbf{1.00} & 0.00 & 0.00 & 0.00 \\
	 & \texttt{presigned-int} & \textbf{1.00} & \textbf{1.00} & \textbf{1.00} & 0.67 & 0.00 & 0.00 & 0.00 \\
	 & \texttt{scattered-url} & 0.67 & 0.00 & \textbf{1.00} & 0.00 & 0.00 & 0.00 & 0.00 \\
	 & \texttt{session-cookie} & \textbf{1.00} & 0.00 & \textbf{1.00} & 0.67 & 0.00 & 0.00 & 0.00 \\
	 & \texttt{sig-canonical} & 0.67 & \textbf{1.00} & \textbf{1.00} & 0.00 & 0.00 & 0.00 & 0.00 \\
	\midrule
	\multicolumn{2}{@{}l|}{\textbf{33-task mean}} & \textbf{0.71} & 0.62 & 0.63 & 0.38 & 0.07 & 0.00 & 0.00 \\
	\bottomrule
\end{tabular}

\end{table}

\section{Per-Label Retry Breakdown}
\label{app:per_label_retry}

The $+7.1$\,pp paired gap reported in \autoref{tab:dose_response} (main body \autoref{sec:labels}) is not uniform across failure labels. \appautoref{tab:per_label_retry} below decomposes the 42 paired GPT-5.1 episodes by the failure label assigned at first attempt, showing where retry-time coaching helps and where it does not. The headline asymmetry --- recovery concentrates in integrity-style labels (\texttt{WRONG\_VALUE}, \texttt{MISSING\_CONSTRAINT}); temporal labels are unresponsive --- and the \texttt{RATE\_LIMITED} inversion are discussed in the main body.

\begin{table}[!htbp]
	\centering\small
	\caption{Per-label retry success on the $n{=}42$ paired GPT-5.1 failures. Cells show pass-count under each retry condition; the right column gives the $\Delta$ between correct- and wrong-label coaching.}
	\label{tab:per_label_retry}
	\begin{tabular}{@{}lrrrrr@{}}
		\toprule
		\textbf{Failure label}        & \textbf{$n$} & \textbf{Naive} & \textbf{Correct-label} & \textbf{Wrong-label} & \textbf{$\Delta$ (correct $-$ wrong)} \\
		\midrule
		\texttt{WRONG\_VALUE}         & 19           & 3              & \textbf{4}             & 2                    & $+2$                                  \\
		\texttt{MISSING\_CONSTRAINT}  & 1            & 0              & \textbf{1}             & 0                    & $+1$                                  \\
		\texttt{EXPIRED\_BEFORE\_USE} & 9            & 2              & 2                      & 1                    & $+1$                                  \\
		\texttt{RATE\_LIMITED}        & 1            & 0              & 0                      & \textbf{1}           & $-1$ (inversion)                      \\
		other (rare labels)           & 12           & 1              & 1                      & 1                    & $\phantom{-}0$                        \\
		\midrule
		\textbf{Total}                & 42           & 6              & \textbf{8}             & 5                    & $+3$ pass-count, $+7.1$\,pp rate      \\
		\bottomrule
	\end{tabular}
\end{table}

\section{Full Failure-Mode Profiles (All Models)}
\label{app:failure_dist_full}

\appautoref{fig:failure_dist} arranges seven complete-coverage models from \appautoref{tab:llm_results} as a capability-ladder pie grid (ordered by SR). Phi-4 14\,B (sub-cliff) fails entirely with \texttt{OTHER} --- a \emph{pre-contract failure} where the agent gives up before reaching any contract predicate, not a taxonomy coverage gap. From there the dominant failure type shifts predictably with capability: integrity (\texttt{WRONG\_VALUE}, red) dominates mid-tier and cliff-emerging models, then temporal (\texttt{EXPIRED\_BEFORE\_USE}, blue) takes over at the frontier (Sonnet 4.5, Opus 4.6); the open-source MoE ceiling (Qwen3.5-397B-A17B) sits at the same transition point. GPT-4o's red-dominant profile is exactly what GPT-5.1 regresses toward in \autoref{sec:version_regression}, motivating the GPT-5.1 test bed used in \autoref{sec:labels}.

\begin{figure}[!htbp]
	\centering
	\includegraphics[width=\textwidth]{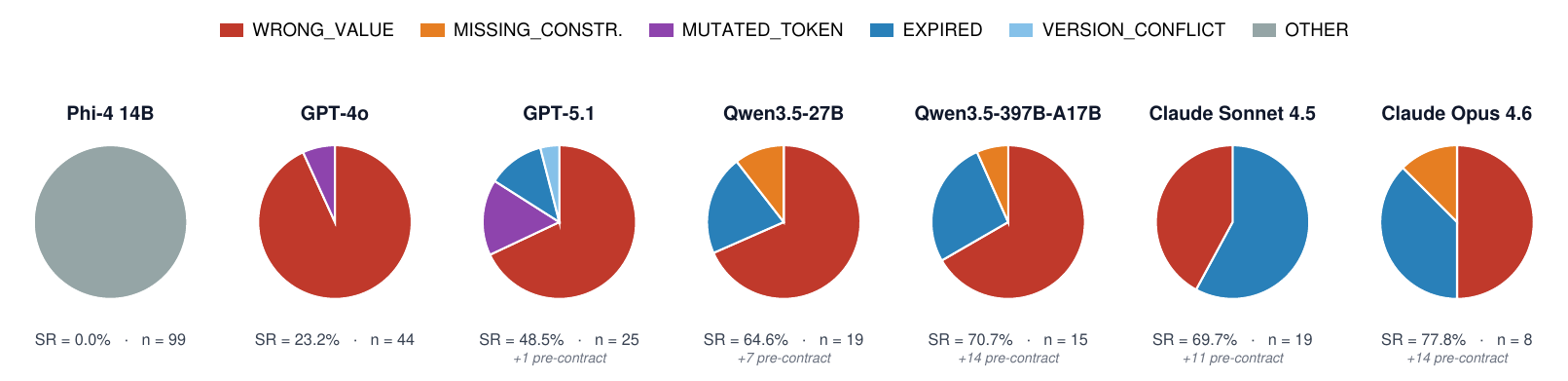}
	\caption{Capability ladder of failure-mode profiles, ordered by SR. The dominant failure type shifts gray (sub-cliff \texttt{OTHER}) $\to$ red (mid-tier integrity) $\to$ blue (frontier temporal).}
	\label{fig:failure_dist}
\end{figure}

\appautoref{fig:failure_dist_full} below complements the pie-grid view with the full per-model stacked-bar profile across all 14 models with $\geq 60$\,episode coverage in our cohort, including the full GPT-5 series (4o / 5 / 5.1 / 5.2) and the open-source SOTA tier. The stacked-bar view exposes per-label proportions as horizontal bands that are easy to scan across models (e.g., the \texttt{WRONG\_VALUE} band visibly grows in GPT-5.1 vs.\ GPT-5, then shrinks in GPT-5.2 --- the V-shape of \autoref{sec:version_regression}). The pie grid sacrifices that horizontal-scan property in exchange for showing the dominant-failure-type progression along the capability axis as a single Gestalt; the two figures answer different questions on the same underlying data.

\begin{figure}[!htbp]
	\centering
	\includegraphics[width=\textwidth]{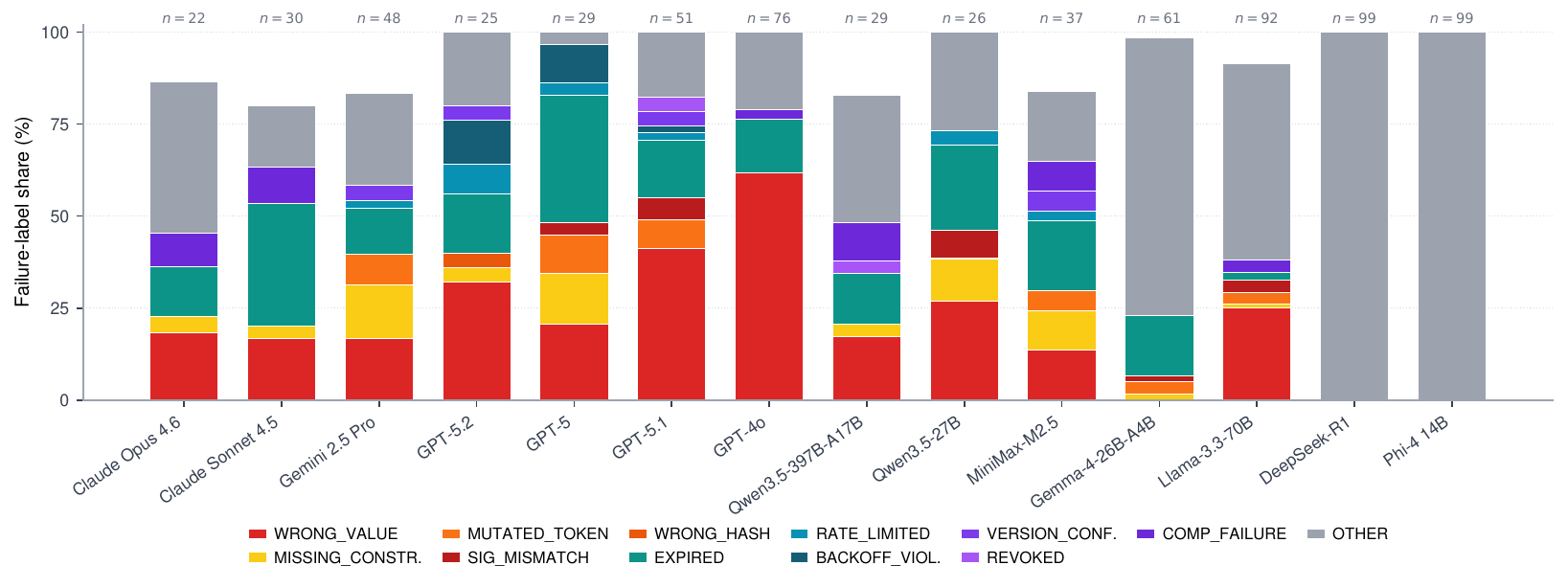}
	\caption{Full per-model failure-label profile (stacked bars, all 15 complete- or near-complete-coverage models). Each bar's segments sum to the model's total failed-episode count (printed as $n{=}\cdots$ above the bar); segment heights are the proportion of failed episodes attributable to each label category. Companion view to the capability-ladder pie grid in \appautoref{fig:failure_dist}.}
	\label{fig:failure_dist_full}
\end{figure}


\end{document}